\documentclass{article}
\usepackage[hscale=0.7,vscale=0.8]{geometry}
\pdfoutput=1
\usepackage{cmap} 

\usepackage{authblk}

\usepackage{setspace} % setspace should be ahead of hyperref
\usepackage{url}
\usepackage{hyperref}
\usepackage[subrefformat=parens,labelformat=parens]{subfig}
\usepackage{multirow}
\usepackage{verbatim}
\usepackage{amsmath}
\usepackage{amsthm}
\usepackage{amssymb}
\usepackage[noend]{algorithmic}
\usepackage{algorithm}
\usepackage{array}
\usepackage{balance}
\usepackage{graphicx}
\usepackage{booktabs}
\usepackage[square,sort,comma,numbers]{natbib}
\usepackage{paralist}
\usepackage{tabularx}
\usepackage{dblfloatfix}
\usepackage{fixltx2e}
\usepackage{color}
\usepackage{textcomp}
\newcommand{\tableref}[1]{Table~\ref{tab:#1}}
\newcommand{\figref}[1]{Figure~\ref{fig:#1}}
\newcommand{\figsubref}[1]{Figure~\subref*{fig:#1}}

\renewcommand\footnoterule{\kern-3pt \hrule width 2in \kern 2.6pt}

\newcolumntype{C}[1]{>{\centering\arraybackslash}p{#1}}

\newtheorem{definition}{Definition}
\newtheorem{theorem}{Theorem}
\newtheorem{lemma}{Lemma}

\title{A Server-based Approach for Predictable GPU Access with Improved Analysis}

\author[1]{Hyoseung Kim}
\affil[1]{University of California, Riverside \\ hyoseung@ucr.edu}
\author[2]{Pratyush Patel}
\affil[2]{Carnegie Mellon University \\ pratyusp@andrew.cmu.edu}
\author[3]{Shige Wang}
\affil[3]{General Motors R\&D \\ shige.wang@gm.com}
\author[4]{Ragunathan (Raj) Rajkumar}
\affil[4]{Carnegie Mellon University \\rajkumar@cmu.edu}

\begin{document}
\date{}
\maketitle

\pagestyle{plain}
\pagenumbering{arabic}

\begin{abstract}

We propose a server-based approach to manage a general-purpose graphics processing unit (GPU) in a predictable and efficient manner. Our proposed approach introduces a GPU server that is a dedicated task to handle GPU requests from other tasks on their behalf. The GPU server ensures bounded time to access the GPU, and allows other tasks to suspend during their GPU computation to save CPU cycles. By doing so, we address the two major limitations of the existing real-time synchronization-based GPU management approach: busy waiting within critical sections and long priority inversion. We have implemented a prototype of the server-based approach on a real embedded platform. This case study demonstrates the practicality and effectiveness of the server-based approach. Experimental results indicate that the server-based approach yields significant improvements in task schedulability over the existing synchronization-based approach in most practical settings. Although we focus on a GPU in this paper, the server-based approach can also be used for other types of computational accelerators. 

%\VS{-4pt}
\end{abstract}

\section{Introduction}
The high computational demands of complex algorithmic tasks used in recent embedded and cyber-physical systems pose substantial challenges in guaranteeing their timeliness. 
For example, a self-driving car~\cite{Kim_ICCPS13,Wei_IV13} executes perception and motion planning algorithms in addition to running tasks for data fusion from tens of sensors equipped within the vehicle. Since these tasks are computationally intensive, it becomes hard to satisfy their timing requirements when they execute on the same hardware platform. 
Fortunately, many of today's embedded multi-core processors, such as NXP i.MX6~\cite{iMX6} and NVIDIA TX1/TX2~\cite{NVIDIA_TX1}, have an on-chip, general-purpose graphics processing unit (GPU), which can greatly help in addressing the timing challenges of computation-intensive tasks by accelerating their execution. 

The use of GPUs in a time predictable manner brings up several challenges. 
First, many of today's commercial-off-the-shelf (COTS) GPUs do not support a preemption mechanism, and GPU access requests from application tasks are handled in a sequential, non-preemptive manner.
This is primarily due to the high overhead expected on GPU context switching~\cite{Tanasic_ISCA14}.  Although some recent GPU architectures, such as NVIDIA Pascal~\cite{NVIDIA_Pascal}, claim to offer GPU preemption, there is no documentation regarding their explicit behavior, and existing drivers (and GPU programming APIs) do not offer any programmer control over GPU preemption at the time of writing this paper. 
Second, COTS GPU device drivers do not respect task priorities and the scheduling policy used in the system. Hence, in the worst case, the GPU access request of the highest-priority task may be delayed by the requests of all lower-priority tasks in the system, which could possibly cause unbounded priority inversion. 

%Today's commercial-off-the-shelf (COTS) GPUs, however, are not designed with predictability as a primary concern~\cite{Elliott_RTSS13}. First of all, execution on a GPU is non-preemptive. While a lower-priority task is using a GPU, GPU execution requests from higher-priority tasks are delayed until the current GPU execution finishes. In addition, GPU device drivers do not consider the scheduling policy used in the system. Hence, GPU requests from lower-priority tasks may be handled earlier than those from higher-priority tasks, which hamper the schedulability of tasks. 
%In this paper, we focus on the problem of utilizing a COTS GPU in a timely and efficient manner.

The aforementioned issues have motivated the development of predictable GPU management techniques to ensure task timing constraints while achieving performance improvement~\cite{Elliott_RTS12,Elliott_RTSS13,Elliott_RTS13,Kato_RTSS11,Kato_ATC11, Kim_RTSS13,Zhou_RTAS15}.
%Many techniques have been developed to utilize a COTS GPU as a predictable, shared computing resource~\cite{Kato_ATC11, Kato_RTSS11, Zhou_RTAS15, Kim_RTSS13, Elliott_RTS12, Elliott_RTSS13}. 
Among them, the work in \cite{Elliott_RTS12,Elliott_RTSS13,Elliott_RTS13} introduces a {\em synchronization-based approach} that models GPUs as mutually-exclusive resources and uses real-time synchronization protocols to arbitrate GPU access. This approach has many benefits. First, it can schedule GPU requests from tasks in an analyzable manner, without making any change to GPU device drivers. %, which are typically closed source. 
Second, it allows the existing task schedulability analysis methods, originally developed for real-time synchronization protocols, to be easily applied to analyze tasks accessing GPUs. However, due to the underlying assumption on critical sections, this approach requires tasks to busy-wait during the entire GPU execution, thereby resulting in substantial CPU utilization loss. Note that semaphore-based locks also experience this problem as the busy waiting occurs {\em within} a critical section after acquiring the lock. Also, the use of real-time synchronization protocols for GPUs may unnecessarily delay the execution of high-priority tasks due to the priority-boosting mechanism employed in some protocols, such as MPCP~\cite{MPCP} and FMLP+~\cite{Brandenburg_11}. We will review these issues in Section~\ref{GPU_limitations_of_synchronization_approach}.

In this paper, we develop a {\em server-based approach} for predictable GPU access control to address the aforementioned limitations of the existing synchronization-based approach. 
Our proposed approach introduces a dedicated GPU server task that receives GPU access requests from other tasks and handles the requests on their behalf. Unlike the synchronization-based approach, the server-based approach allows tasks to suspend during GPU computation while preserving analyzability. This not only yields CPU utilization benefits, but also reduces task response times. We present the schedulability analysis of tasks under our server-based approach, which accounts for the overhead of the GPU server. Although we have focused on a GPU in this work, our approach can be used for other types of computational accelerators, such as a digital signal processor (DSP). 

We have implemented a prototype of our approach on a SABRE Lite embedded platform~\cite{SabreLite} equipped with four ARM Cortex-A9 CPUs and one Vivante GC2000 GPU. Our case study using this implementation with the workzone recognition algorithm~\cite{Lee_ITSC13} developed for a self-driving car demonstrates the practicality and effectiveness of our approach in improving CPU utilization and reducing response time. We have also conducted detailed experiments on task schedulability. Experimental results show that while our server-based approach does not dominate the synchronization-based approach, it outperforms the latter in most of the practical cases. 

This paper is an extended version of our conference paper~\cite{Kim_RTCSA17}, with the following new contributions: (i) an improved analysis for GPU request handling time under the server-based approach, (ii) a discussion on task allocation with the GPU server, and (iii) new experimental results with the improved analysis.

The rest of this paper is organized as follows: Section~\ref{related_work} reviews relevant prior work. Section~\ref{system_model} describes our system model. Section~\ref{gpu_synchronization_approach} reviews the use of the synchronization-based approach for GPU access control and discusses its limitations. Section~\ref{gpu_server_approach} presents our proposed server-based approach. Section~\ref{evaluation} evaluates the approach using a practical case study and overhead measurements, along with detailed schedulabilty experiments. Section~\ref{conclusions} concludes the paper.

\section{Related Work}
\label{related_work}
Many techniques have been developed to utilize a GPU as a predictable, shared computing resource.
TimeGraph~\cite{Kato_ATC11} is a real-time GPU scheduler that schedules GPU access requests from tasks with respect to task priorities. This is done by modifying an open-source GPU device driver and monitoring GPU commands at the driver level. 
%TimeGraph also provides a resource reservation mechanism that accounts for and enforces the GPU usage of each task, with posterior and apriori enforcement techniques. 
RGEM~\cite{Kato_RTSS11} allows splitting a long data-copy operation into smaller chunks, reducing blocking time on data-copy operations. Gdev~\cite{Kato_ATC12} provides common APIs to both user-level tasks and the OS kernel to use a GPU as a standard computing resource. GPES~\cite{Zhou_RTAS15} is a software technique to break a long GPU execution segment into smaller sub-segments, allowing preemptions at the boundaries of sub-segments. 
While all these techniques can mitigate the limitations of today's GPU hardware and device drivers, they have not considered the schedulability of tasks using the GPU. In other words, they handle GPU requests from tasks in a predictable manner, but do not formally analyze the worst-case timing behavior of tasks on the CPU side, which is addressed in this paper. 

Elliott et al.~\cite{Elliott_RTS12,Elliott_RTSS13,Elliott_RTS13} modeled GPUs as mutually-exclusive resources and proposed the use of real-time synchronization protocols for accessing GPUs. Based on this, they developed GPUSync~\cite{Elliott_RTSS13}, a software framework for GPU management in multi-core real-time systems. GPUSync supports both fixed- and dynamic-priority scheduling policies, and provides various features, such as budget enforcement, multi-GPU support, and clustered scheduling. It uses separate locks for copy and execution engines of GPUs to enable overlapping of GPU data transmission and computation.
The pros and cons of the synchronization-based approach in general will be thoroughly discussed in Section~\ref{gpu_synchronization_approach}.

Server-based software architectures for GPU access control have been studied in the context of GPU virtualization, such as rCUDA~\cite{Duato_HCPS10} and vCUDA~\cite{Shi_IEEE12}. They implement an RPC server to receive GPU acceleration requests from high-performance computing applications running on remote or virtual machines. The received requests are then handled by the server on a host machine equipped with a GPU. While our work uses a similar software architecture to these approaches, our work differs by providing guaranteed bounds on the worst-case GPU access time.

The self-suspension behavior of tasks has been studied in the context of real-time systems~\cite{Audsley_RTAS04,Bletsas_TR15,Chen_TR16,Kim_RTSS13}. This is motivated by the fact that tasks can suspend while accessing hardware accelerators like GPUs. Kim et al.~\cite{Kim_RTSS13} proposed {\em segment-fixed priority scheduling}, which assigns different priorities and phase offsets to each segment of tasks. They developed several heuristics for priority and offset assignment because finding the optimal solution for that assignment is NP-hard in the strong sense. Chen et al.~\cite{Chen_TR16} reported errors in existing self-suspension analyses and presented corrections for the errors. Those approaches assume that the duration of self-suspension is given as a fixed task parameter. However, this assumption does not comply with the case where a task accesses a shared GPU and the waiting time for the GPU is affected by other tasks in the system. In this work, we use the results in \cite{Bletsas_TR15,Chen_TR16} to take into account the effect of self-suspension in task schedulability, and propose techniques to bound the worst-case access time to a shared GPU.

Recently, there have been efforts to understand the details of GPU-internal scheduling on some NVIDIA GPUs. Otterness et al.~\cite{Otterness_RTAS17} report that multiple GPU execution requests from different processes may be scheduled concurrently on a single NVIDIA TX1 GPU but the execution time of individual requests becomes longer and less predictive compared to when they run sequentially. Amert et al.~\cite{Amert_RTSS17} discuss the internal scheduling policy of NVIDIA TX2 for GPU execution requests made by threads sharing the same address space, i.e., those belonging to the same process. While these efforts have potential to improve GPU utilization and overall system efficiency, it is unclear whether their findings are valid on other GPU architectures and other types of accelerators. Although we do not consider these recently-found results in this work as our goal is to develop a broadly-applicable solution, we plan to consider them in our future work.
%Hence, we do not consider them in this work as our goal is to develop a widely-applicable solution.

\begin{figure}[t]
	\centering
	\subfloat{
		\includegraphics[width=0.8\columnwidth]{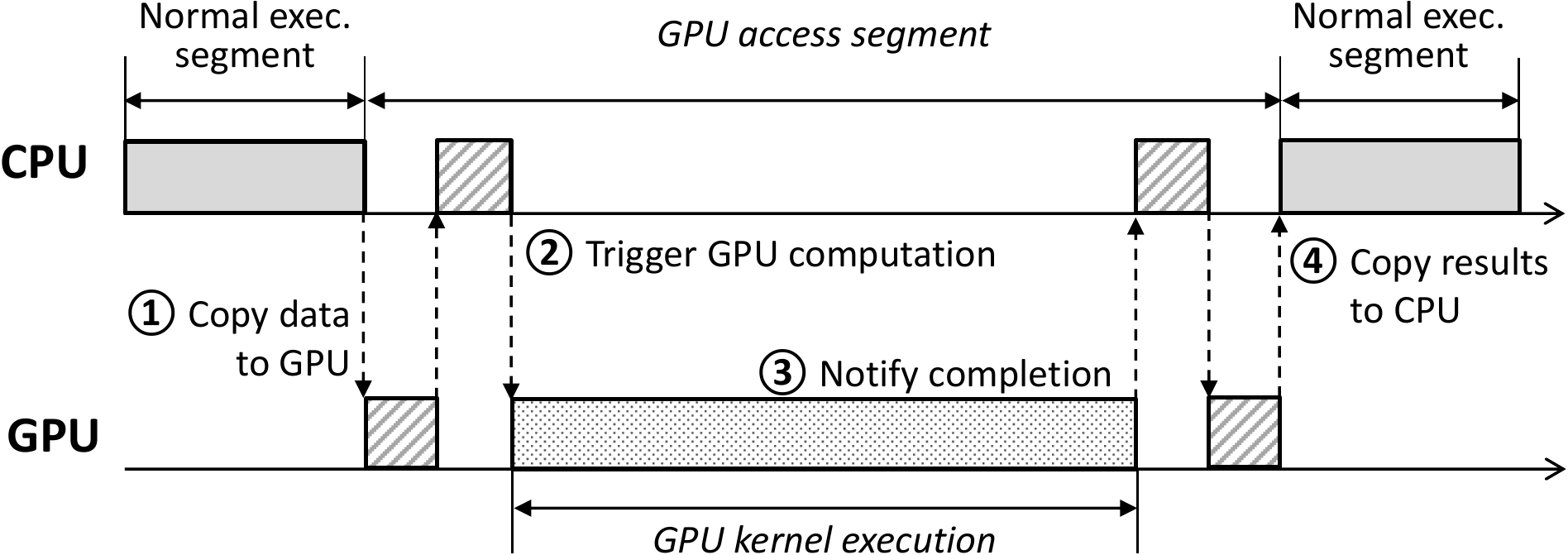}
	}
	\caption{Execution pattern of a task accessing a GPU}
	\label{fig:GPU_execution_pattern}	
\end{figure}

\section{System Model}
\label{system_model}
The work in this paper assumes a multi-core platform equipped with a single general-purpose GPU device.\footnote{This assumption reflects today's GPU-enabled embedded processors, e.g., NXP i.MX6~\cite{iMX6} and NVIDIA TX1/TX2~\cite{NVIDIA_TX1}. A multi-GPU platform would be used for future real-time embedded and cyber-physical systems, but extending our work to handle multiple GPUs remains as future work.}
The GPU is shared among multiple tasks, and GPU requests from tasks are handled in a sequential, non-preemptive manner. The GPU has its own memory region, which is assumed to be sufficient enough for the tasks under consideration. We do not assume the concurrent execution of GPU requests from different tasks, called {\em GPU co-scheduling}, because recent work~\cite{Otterness_RTAS17} reports that ``co-scheduled GPU programs from different programs are not truly concurrent, but are multiprogrammed instead'' and ``this (co-scheduling) may lead to slower or less predictable total times in individual programs''.

We consider sporadic tasks with constrained deadlines. The execution time of a task using a GPU is decomposed into {\em normal execution segments} and {\em GPU access segments}. Normal execution segments run entirely on CPU cores and GPU access segments involve GPU operations. \figref{GPU_execution_pattern} depicts an example of a task having a single GPU access segment. In the GPU access segment, the task first copies data needed for GPU computation, from CPU memory to GPU memory (Step~\textcircled{\footnotesize 1} in \figref{GPU_execution_pattern}). This is typically done using Direct Memory Access (DMA), which requires no (or minimal) CPU intervention. 
The task then triggers the actual GPU computation, also referred to as {\em GPU kernel execution}, and waits for the GPU computation to finish (Step~\textcircled{\footnotesize 2}). The task is notified when the GPU computation finishes (Step~\textcircled{\footnotesize 3}), and it copies the results back from the GPU to the CPU (Step~\textcircled{\footnotesize 4}). Finally, the task continues its normal execution segment. Note that during the time when CPU intervention is not required, e.g., during data copies with DMA and GPU kernel execution, the task may suspend or busy-wait, depending on the implementation of the GPU device driver and the configuration used. 

\smallskip
\noindent\textbf{Synchronous and Asynchronous GPU Access.}
The example in \figref{GPU_execution_pattern} uses {\em synchronous mode} for GPU access, where each GPU command, such as memory copy and GPU kernel execution, can be issued only after the prior GPU command has finished. However, many GPU programming interfaces, such as CUDA and OpenCL, also provide {\em asynchronous mode}, which allows a task to overlap CPU and GPU computations. For example, if a task sends all GPU commands in asynchronous mode at the beginning of the GPU access segment, then it can either perform other CPU operations within the same GPU access segment or simply wait until all the GPU commands complete. Hence, while the sequence of GPU access remains the same, the use of asynchronous mode can affect the amount of active CPU time in a GPU segment.

\smallskip
\noindent\textbf{Task Model.}
A task $\tau_i$ is characterized as follows:
$${\tau}_i:=(C_i, T_i, D_i, G_i, \eta_i)$$
\begin{itemize}
	\item $C_i$: the sum of the worst-case execution time (WCET) of all normal execution segments of task ${\tau}_i$.
	\item $T_i$: the minimum inter-arrival time of each job of ${\tau_i}$.
	\item $D_i$: the relative deadline of each job of $\tau_i$ ($D_i \le T_i$).
	\item $G_i$: the maximum accumulated duration of all GPU segments of $\tau_i$, when there is no other task competing for a GPU. 
	\item $\eta_i$: the number of GPU access segments in each job of $\tau_i$.
\end{itemize}
The utilization of a task $\tau_i$ is defined as $U_i=(C_i+G_i)/T_i$. Parameters $C_i$ and $G_i$ can be obtained by either measurement-based or static-analysis tools. If a measurement-based approach is used, $C_i$ and $G_i$ can be measured when $\tau_i$ executes alone in the system without any other interfering tasks and they need to be conservatively estimated. A task using a GPU has one or more GPU access segments. We use $G_{i,j}$ to denote the maximum duration of the $j$-th GPU access segment of $\tau_i$, i.e., $G_i=\sum_{j=1}^{\eta_i}G_{i,j}$.
Parameter $G_{i,j}$ can be decomposed as follows:
$$G_{i,j}:=(G^e_{i,j}, G^m_{i,j})$$
\begin{itemize}
	\item $G^e_{i,j}$: the WCET of pure GPU operations that do not require CPU intervention in the $j$-th GPU access segment of $\tau_i$. 
	\item $G^m_{i,j}$: the WCET of miscellaneous operations that require CPU intervention in the $j$-th GPU access segment of $\tau_i$.
\end{itemize}
$G^e_{i,j}$ includes the time for GPU kernel execution, and $G^m_{i,j}$ includes the time for copying data, launching the kernel, notifying the completion of GPU commands, and executing other CPU operations. The cost of triggering self-suspension during CPU-inactive time in a GPU segment is assumed to be taken into account by $G^m_{i,j}$. If data are copied to or from the GPU using DMA, only the time for issuing the copy command is included in $G^m_{i,j}$; the time for actual data transmission by DMA is modeled as part of $G^e_{i,j}$, as it does not require CPU intervention.
Note that $G_{i,j} \le G^e_{i,j}+G^m_{i,j}$ because $G^e_{i,j}$ and $G^m_{i,j}$ are not necessarily observed on the same control path and they may overlap in asynchronous mode.

\smallskip
\noindent\textbf{CPU Scheduling.}
In this work, we focus on {\em partitioned fixed-priority preemptive task scheduling} due to the following reasons: \emph{(i)}~it is widely supported in many commercial real-time embedded OSs such as OKL4~\cite{OKL4} and QNX RTOS~\cite{QNX_RTOS}, and \emph{(ii)}~it does not introduce task migration costs. 
%, and (iii) it can benefit from the well-established uniprocessor theoretical framework.
%Tasks are scheduled by {\em partitioned fixed-priority preemptive scheduling}.
Thus, each task is statically assigned to a single CPU core. Any fixed-priority assignment, such as Rate-Monotonic~\cite{Liu_Layland} can be used for tasks. 
Each task $\tau_i$ is assumed to have a unique priority $\pi_i$. An arbitrary tie-breaking rule can be used to achieve this assumption under fixed-priority scheduling.

\section{Limitations of Synchronization-based GPU Access Control}
\label{gpu_synchronization_approach}

In this section, we characterize the limitations of using a real-time synchronization protocol for tasks accessing a GPU on a multi-core platform. We consider semaphore-based synchronization, where a task attempting to acquire a lock suspends if the lock is already held by another task. This is because semaphores are more efficient than spinlocks for long critical sections~\cite{Lakshmanan_RTSS09}, which is the case for GPU access.

\subsection{Overview}
The synchronization-based approach models the GPU as a mutually-exclusive resource and the GPU access segments of tasks as critical sections. A single semaphore-based mutex is used for protecting such GPU critical sections. Hence, under the synchronization-based approach, a task should hold the GPU mutex to enter its GPU access segment. 
%The task releases the mutex when it leaves the GPU access segment. a task can only enter its GPU access segment when the mutex for the GPU is not held by any other task. 
If the mutex is already held by another task, the task is inserted into the waiting list of the mutex and suspends until the mutex can be held by that task. Some implementations like GPUSync~\cite{Elliott_RTSS13} use separate locks for internal resources of the GPU, e.g., copy and execution engines, to achieve parallelism in GPU access, but we focus on a single lock for the entire GPU for simplicity. 

Among a variety of real-time synchronization protocols, we consider the Multiprocessor Priority Ceiling Protocol (MPCP)~\cite{MPCP2,MPCP} as a representative, because it works for partitioned fixed-priority scheduling that we use in our work and it has been widely referenced in the literature. 
We shall briefly review the definition of MPCP below. More details on MPCP can be found in \cite{Lakshmanan_ECRTS09,MPCP2,MPCP}. 
\begin{enumerate}
	%	\item Each lock protecting a global resource uses a priority queue for its waiting list. 
	%each global resource $R_k$ is associated with the two types of priority ceilings: {\em task-level} and {\em VCPU-level} priority ceilings.
	\item When a task $\tau_i$ requests an access to a resource $R_k$, it can be granted to $\tau_i$ if it is not held by another task.
	\item While a task $\tau_i$ is holding a resource that is shared among tasks assigned to different cores, the priority of $\tau_i$ is raised to $\pi_{B}+\pi_{i}$, where $\pi_{B}$ is a base task-priority level greater than that of any task in the system, and $\pi_{i}$ is the normal priority of $\tau_i$. This ``priority boosting'' under MPCP is referred to as the {\em global priority ceiling} of $\tau_i$.
	\item When a task $\tau_i$ requests access to a resource $R_k$ that is already held by another task, $\tau_i$ is inserted to the waiting list of the mutex for $R_k$ and suspends. 
	\item When a resource $R_k$ is released and the waiting list of the mutex for $R_k$ is not empty, the highest-priority task in the waiting list is dequeued and granted $R_k$. 
\end{enumerate}
%See Appendix for the task schedulability analysis under the synchronization-based approach with MPCP. \note{RTCSA doesn't allow past 10 pages as per the website. Will/should we include appendix?}
%The analysis presented in the Appendix is originally developed by Lakshmanan et al.~\cite{Lakshmanan_RTSS09}, and combined with a correction given by Chen et al.~\cite{Chen_TR16}.

\subsection{Limitations}
\label{GPU_limitations_of_synchronization_approach}

As described in Section~\ref{system_model}, each GPU access segment contains various operations, including data copies, notifications, and GPU computation. Specifically, a task may suspend when CPU intervention is not required, e.g., during GPU kernel execution, to save CPU cycles. However, under the synchronization-based approach, any task in its GPU access segment (i.e., running within a critical section after acquiring the lock) should ``busy-wait'' for any operation conducted on the GPU in order to ensure timing predictability. This is because most real-time synchronization protocols and their analyses, such as MPCP~\cite{MPCP}, FMLP~\cite{Block_RTCSA07}, FMLP+~\cite{Brandenburg_11} and OMLP~\cite{Brandenburg_13}, assume that \emph{(i)} a critical section is executed entirely on the CPU, and \emph{(ii)} there is no suspension during the execution of the critical section. Hence, during its GPU kernel execution, the task is not allowed to suspend even when no CPU intervention is required.\footnote{These protocols allow suspension while waiting for lock acquisition.} For example, while the GPUSync implementation~\cite{Elliott_RTSS13} can be configured to suspend instead of busy-wait during GPU kernel execution at runtime, its analysis uses a suspension-oblivious approach that does not incorporate benefits from self-suspension during GPU execution. 
%\footnote{The GPUSync implementation~\cite{Elliott_RTSS13} can be configured to suspend instead of busy-wait during GPU kernel execution, but it uses suspension-oblivious analysis, which does not capture CPU time saving from self-suspension.}
As the time for GPU kernel execution and data transmission by DMA increases, the CPU time loss under the synchronization-based approach is therefore expected to increase. The analyses of those protocols could possibly be modified to allow suspension within critical sections, at the potential expense of increased pessimism.\footnote{The analyses for MPCP~\cite{Lakshmanan_RTSS09} and FMLP+~\cite{Brandenburg_11} under partitioned scheduling show that, whenever a task resumes from suspension, it may experience priority inversion from local tasks with higher boosted priorities. Hence, allowing self-suspension within critical sections under these protocols will likely increase chances of priority inversion and requires a change in the task model to limit the number of suspensions. Note that this however does not affect the asymptotic optimality of FMLP+, as reported in~\cite{Brandenburg_ECRTS14}.}

\begin{figure}[t]
	\centering
	\subfloat{
		\includegraphics[width=0.8\columnwidth]{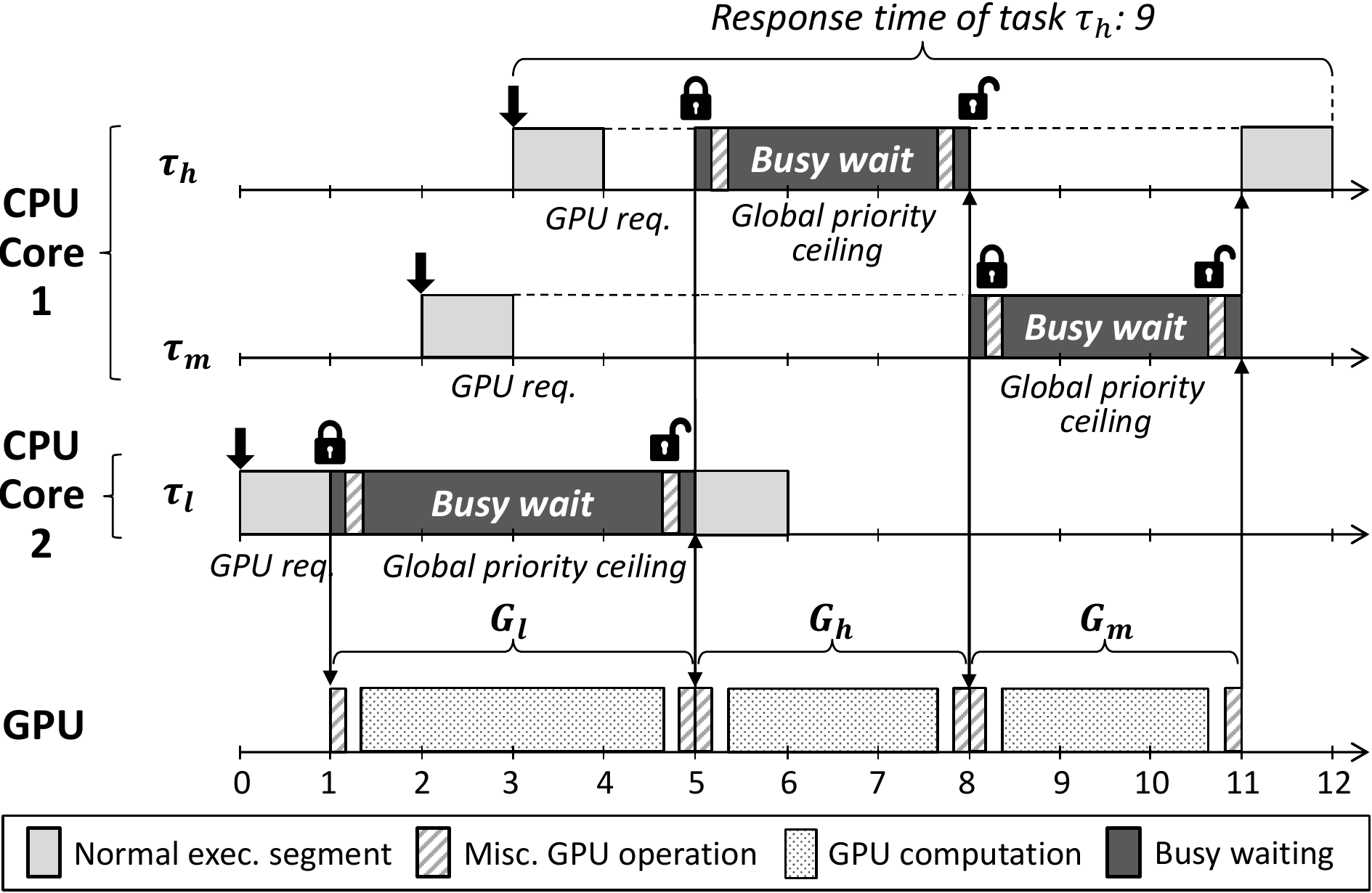}
	}
	\caption{Example schedule of GPU-using tasks under the synchronization-based approach with MPCP}
	\label{fig:GPU_synch_based_approach_example}
\end{figure}

Some synchronization protocols, such as MPCP~\cite{MPCP}, FMLP~\cite{Block_RTCSA07} and FMLP+~\cite{Brandenburg_ECRTS14}, use priority boosting (either restricted or unrestricted) as a progress mechanism to prevent unbounded priority inversion. However, the use of priority boosting could cause another problem we call ``long priority inversion''. 
We describe this problem with the example in \figref{GPU_synch_based_approach_example}. There are three tasks, $\tau_h$, $\tau_m$, and $\tau_l$, that have high, medium, and low priorities, respectively. 
Each task has one GPU access segment that is protected by MPCP and executed between two normal execution segments. $\tau_h$ and $\tau_m$ are allocated to Core 1, and $\tau_l$ is allocated to Core 2. 

Task $\tau_l$ is released at time 0 and makes a GPU request at time 1. Since there is no other task using the GPU at that point, $\tau_l$ acquires the mutex for the GPU and enters its GPU access segment. $\tau_l$ then executes with the global priority ceiling associated with the mutex (priority boosting). 
Note that while the GPU kernel of $\tau_l$ is executed, $\tau_l$ also consumes CPU cycles due to the busy-waiting requirement of the synchronization-based approach.
Tasks $\tau_m$ and $\tau_h$ are released at time 2 and 3, respectively. They make GPU requests at time 3 and 4, but the GPU cannot be granted to either of them because it is already held by $\tau_l$. Hence, $\tau_m$ and $\tau_h$ suspend. At time 5, $\tau_l$ releases the GPU mutex, and $\tau_h$ acquires the mutex next because it has higher priority than $\tau_m$. At time 8, $\tau_h$ finishes its GPU segment, releases the mutex, and restores its normal priority. Next, $\tau_m$ acquires the mutex, and preempts the normal execution segment of $\tau_h$ because the global priority ceiling of the mutex is higher than $\tau_h$'s normal priority. Hence, although the majority of $\tau_m$'s GPU segment merely performs busy-waiting, the execution of the normal segment of $\tau_h$ is delayed until the GPU access segment of $\tau_m$ finishes. Finally, $\tau_h$ completes its normal execution segment at time 12, making the response time of $\tau_h$ 9 in this example.

\section{Server-based GPU Access Control}
\label{gpu_server_approach}

\begin{figure}[t]
	\centering
	\subfloat{
		\includegraphics[width=0.8\columnwidth]{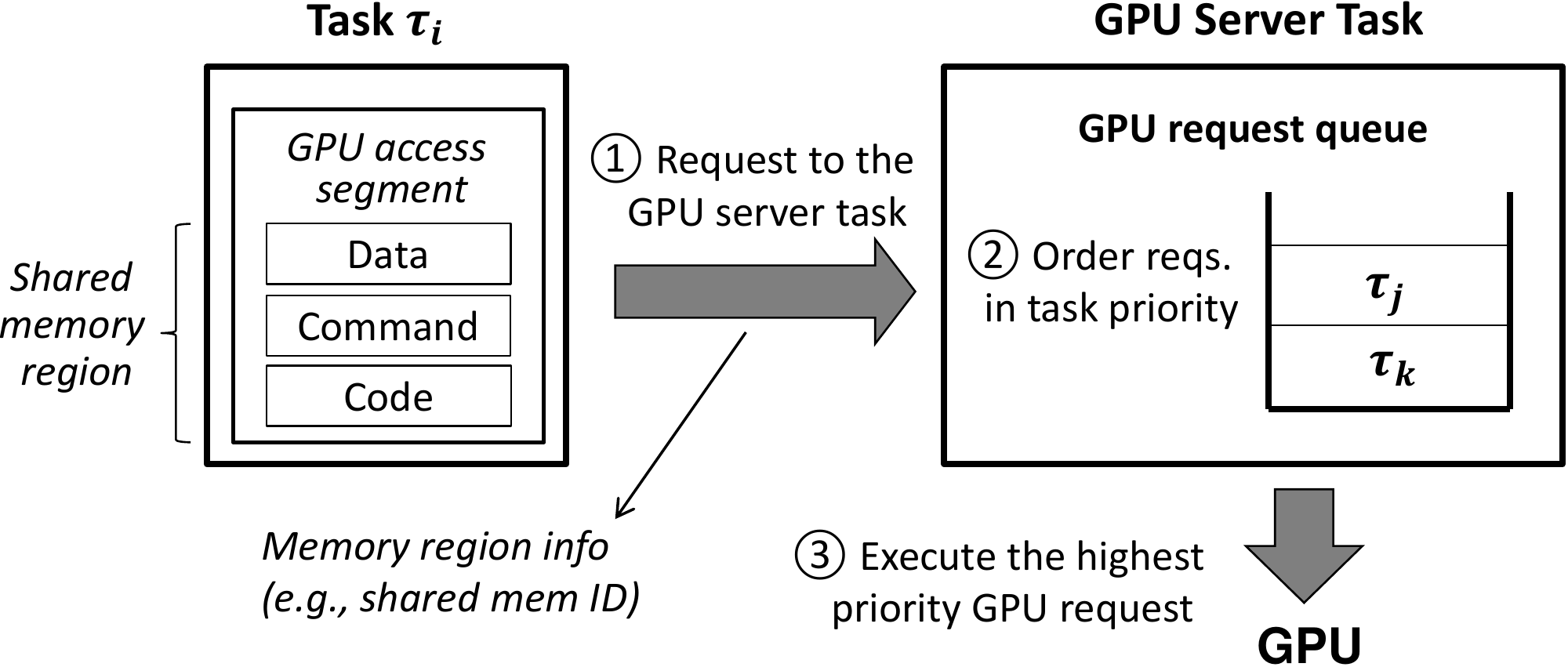}
	}
	\caption{GPU access procedure under our server-based approach}
	\label{fig:GPU_server_based_approach}
\end{figure}

We present our server-based approach for predictable GPU access control. This approach addresses the two main limitations of the synchronization-based approach, namely, busy waiting and long priority inversion.

\subsection{GPU Server Design}
Our server-based approach creates a {\em GPU server} task that handles GPU access requests from other tasks on their behalf. The GPU server is assigned the highest priority in the system, which is to prevent preemptions by other tasks. \figref{GPU_server_based_approach} shows the sequence of GPU request handling under our server-based approach. First, when a task $\tau_i$ enters its GPU access segment, it makes a GPU access request to the GPU server, not to the GPU device driver. The request is made by sending the memory region information for the GPU access segment, including input/output data, commands and code for GPU kernel execution to the server task. This requires the memory regions to be configured as shared regions so that the GPU server can access them with their identifiers, e.g., \texttt{shmid}. After sending the request to the server, $\tau_i$ suspends, allowing other tasks to execute. Secondly, the server enqueues the received request into the GPU request queue, if the GPU is being used by another request. The GPU request queue is a priority queue, where elements are ordered in their task priorities. Thirdly, once the GPU becomes free, the server dequeues a request from the head of the queue and executes the corresponding GPU segment. During CPU-inactive time, e.g., data copy with DMA and GPU kernel execution, the server suspends to save CPU cycles in synchronous mode. In asynchronous mode, the server performs remaining miscellaneous CPU operations of the request being handled, i.e., asynchronous execution part, and suspends after completing them. This approach allows the GPU server to offer the same level of parallelism as in the asynchronous mode of the synchronization-based approach.\footnote{Note that achieving parallelism beyond this level is not in the scope of this paper. For example, if one wants to develop a task utilizing more than one CPU cores during its GPU execution, the corresponding GPU segment needs to be programmed with multithreads, which is different from our task model and needs additional research efforts to guarantee real-time predictability.}
When the request finishes, the server notifies the completion of the request and wakes up the task $\tau_i$. Finally, $\tau_i$ resumes its execution. 

The use of the GPU server task inevitably introduces the following additional computational costs: \emph{(i)}~sending a GPU request to the server and waking up the server task, \emph{(ii)}~enqueueing the request and checking the request queue to find the highest-priority GPU request, and \emph{(iii)}~notifying the completion of the request to the corresponding task. We use the term $\epsilon$ to characterize the GPU server overhead that upper-bounds these computational costs.

\begin{figure}[t]
	\centering
	\subfloat{
		\includegraphics[width=0.8\columnwidth]{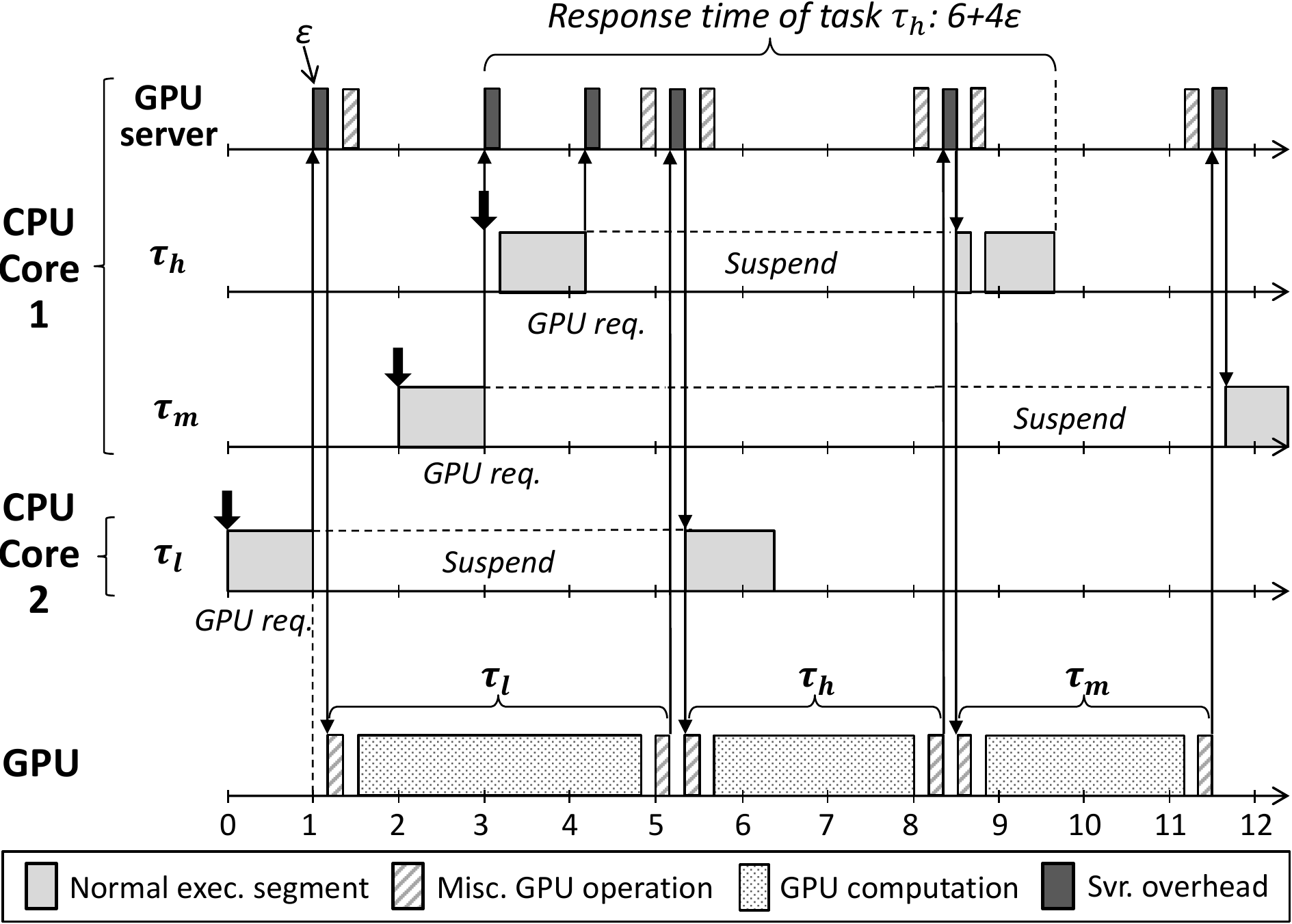}
	}
	\caption{Example schedule under our server-based approach}
	\label{fig:GPU_server_based_approach_example}
\end{figure}

\figref{GPU_server_based_approach_example} shows an example of task scheduling under our server-based approach. This example has the same configuration as the one in \figref{GPU_synch_based_approach_example} and uses synchronous mode for GPU access.\footnote{The GPU server supports both synchronous and asynchronous GPU access.} The GPU server, which the server-based approach creates, is allocated to Core 1. 
%The GPU server overhead $\epsilon$ is assumed to be $1/6$ time units in this example (see Section~\ref{case_study} for $\epsilon$ on a real platform). 
Each GPU access segment has two sub-segments of miscellaneous operations, each of which is assumed to amount to $\epsilon$. 

At time 1, the task $\tau_l$ makes a GPU access request to the server task. The server receives the request and executes the corresponding GPU access segment at time $1+\epsilon$, where $\epsilon$ represents the overhead of the GPU server. Since the server-based approach does not require tasks to busy-wait, $\tau_l$ suspends until the completion of its GPU request. The GPU request of $\tau_m$ at time 3 is enqueued into the request queue of the server. As the server executes with the highest priority in the system, it delays the execution of $\tau_h$ released at time 3 by $\epsilon$. Hence, $\tau_h$ starts execution at time $3+\epsilon$ and makes a GPU request at time $4+\epsilon$. When the GPU access segment of $\tau_l$ completes, the server task is notified. The server then notifies the completion of the GPU request to $\tau_l$ and wakes it up, and subsequently executes the GPU access segment of $\tau_h$ at time $5+2\epsilon$. The task $\tau_h$ suspends until its GPU request finishes. The GPU access segment of $\tau_h$ finishes at time $8+2\epsilon$ and that of $\tau_m$ starts at time $8+3\epsilon$. Unlike the case under the synchronization-based approach, $\tau_h$ can continue to execute its normal execution segment from time $8+3\epsilon$, because $\tau_m$ suspends and the priority of $\tau_m$ is not boosted. The task $\tau_h$ finishes its normal execution segment at time $9+4\epsilon$, and hence, the response time of $\tau_h$ is $6+4\epsilon$. Recall that the response time of $\tau_h$ is 9 for the same taskset under the synchronization-based approach, as shown in \figref{GPU_synch_based_approach_example}. Therefore, we can conclude that the server-based approach provides a shorter response time than the synchronization-based approach for this example taskset, if the value of $\epsilon$ is under $3/4$ time units, which, as measured in Section~\ref{case_study}, is a very pessimistic value for $\epsilon$.

%\note{Ease of implementation? Possible to implement with existing POSIX APIs?}

\subsection{Schedulability Analysis}

We analyze task schedulability under our server-based approach. 
Since all the GPU requests of tasks are handled by the GPU server, we first identify the GPU request handling time for the server. The following properties hold for the server-based approach:

\begin{lemma}
	\label{lemma:server_overhead}
	The GPU server task imposes up to $2\epsilon$ of extra CPU time on each GPU request.
\end{lemma}
\begin{proof}
As the GPU server intervenes before and after the execution of each GPU segment, each GPU request can cause at most $2\epsilon$ of overhead in the worst case. Note that the cost of issuing the GPU request as well as triggering self-suspension within the $j$-th GPU segment of $\tau_i$ is already taken into account by $G_{i,j}^m$, as described in Section~\ref{system_model}.
\end{proof}

When a task $\tau_i$ makes a GPU request, the GPU may be already handling a request from another task. Since the GPU executes in a non-preemptive manner, $\tau_i$'s request must wait for the completion of the currently-handled request. There may also be multiple pending requests that are prioritized over $\tau_i$'s request, and it has to wait for the completion of such requests. As a result, $\tau_i$ experiences {\em waiting time} for GPU access, which we define as follows:
\begin{definition}
	\label{def:waiting_time}
	The waiting time for a GPU access segment of a task $\tau_i$ is the interval between the time when that segment is released and the time when it begins execution.
\end{definition}

\begin{lemma}
	\label{lemma:GPU_server_handling_delay}
	The maximum handling time of all GPU requests of a single job of a task $\tau_i$ by the GPU server is given as follows:
	\begin{equation} \label{eq:GPU_server_handling_delay}
	\begin{split}
	B_i^{gpu}=\left\{
	\begin{array}{lr}
	B_i^{w}+G_i+2\eta_i\epsilon &: \eta_i > 0\\
	0 &: \eta_i = 0\\
	\end{array}
	\right.
	%B_i^{gpu}=\sum_{1\le j\le \eta_i} (B_{i,j}^w+G_{i,j}+2\epsilon)
	\end{split}
	\end{equation}
	where $\eta_i$ is the number of GPU access segments of $\tau_i$ and $B_i^{w}$ is an upper bound on the total waiting time for all $\eta_i$ GPU segments of $\tau_i$.%, i.e., $B_i^{w}=\sum_{1\le j\le \eta_i} B_{i,j}^{w}$. 
\end{lemma}
\begin{proof}
	If $\eta_i>0$, the $j$-th GPU request of $\tau_i$ is handled after some waiting time, and then takes $G_{i,j}+2\epsilon$ to complete the execution (by Lemma~\ref{lemma:server_overhead}). Hence, the maximum handling time of all GPU requests of $\tau_i$ is $B_{i}^w+\sum_{j=1}^{\eta_i} (G_{i,j}+2\epsilon)=B_i^w+G_i+2\eta_i\epsilon$.
	If $\eta_i=0$, $B_i^{gpu}$ is obviously zero. 
\end{proof}

In order to tightly upper-bound the total waiting time $B_i^w$, we adopt the {\em double-bounding approach} used in \cite{Kim_RTS16,Kim_RTAS14}. The double-bounding approach consists of two analysis methods: {\em request-driven} and {\em job-driven}. The request-driven analysis focuses on the worst-case waiting time that can be imposed on each request of $\tau_i$, and computes an upper bound on the total waiting time by adding up all the per-request waiting time. 
%On the other hand, the job-driven analysis upper bounds the total waiting time by capturing the maximum amount of delay that can be generated by other interfering tasks during $\tau_i$'s job execution. 
On the other hand, the job-driven analysis focuses on the amount of waiting time that can be possibly given by other interfering tasks during $\tau_i$'s job execution, and takes the maximum  of this value as an upper bound of the total waiting time.
These two analyses do not dominate one another. Hence, the double-bounding approach takes the minimum of the two to provide a tighter upper bound. The total waiting time $B_i^w$ for all GPU requests of a task $\tau_i$ is therefore:
\begin{equation} \label{eq:GPU_server_waiting_delay}
B_{i}^{w}=\min(B_{i}^{rd}, B_i^{jd})
\end{equation}
where $B_{i}^{rd}$ and $B_{i}^{jd}$ are the results obtained by the request-driven and job-driven analyses, respectively.\footnote{Our initial work in \cite{Kim_RTCSA17} uses only the request-driven analysis to bound the total waiting time, i.e., $B_i^w=B_i^{rd}$.} Note that $B_i^{rd}$ is the sum of the maximum waiting time for each GPU request of $\tau_i$, i.e., $B_i^{rd}=\sum_{1\le j \le \eta_i}B_{i,j}^{rd}$.

\begin{lemma}
	\label{lemma:GPU_server_waiting_delay_rd}
	{\em [Request-driven analysis]} The maximum waiting time for the $j$-th GPU segment of a task $\tau_i$ under the server-based approach is bounded by the following recurrence:
	\begin{equation} \label{eq:GPU_server_waiting_delay_rd}
	\begin{split}
	B_{i,j}^{rd,n+1}&=\max_{\pi_l<\pi_i\land 1\le k\le \eta_l}(G_{l,k}+\epsilon) 
	+ \sum_{\pi_h > \pi_i\land 1\le k \le \eta_h}\left(\left\lceil{B_{i,j}^{rd,n}\over T_h}\right\rceil +1 \right) (G_{h,k}+\epsilon)
	\end{split}
	\end{equation}
	where $G_{l,k}$ is the maximum duration of the $k$-th GPU segment of $\tau_l$ and $B_{i,j}^{rd,0}=\max_{\pi_l<\pi_i\land 1\le k\le \eta_l}(G_{l,k}+\epsilon)$ (the first term of the equation). 
\end{lemma}
\begin{proof}
	When $\tau_i$, the task under analysis, makes a GPU request, the GPU may be handling a request from a lower-priority task, which $\tau_i$ has to wait for completion due to the non-preemptive nature of the GPU. Hence, the longest GPU access segment from all lower-priority tasks needs to be considered as the waiting time in the worst case. Here, only one $\epsilon$ of overhead is caused by the GPU server because other GPU requests will be immediately followed and the GPU server needs to be invoked only once between two consecutive GPU requests, as depicted in \figref{GPU_server_based_approach_example}. This is captured by the first term of the equation. 
	
	During the waiting time of $B_{i,j}^{rd}$, higher-priority tasks can make GPU requests to the server. As there can be at most one carry-in request from each higher-priority task $\tau_h$ during $B_{i,j}^{rd}$, the maximum number of GPU requests made by $\tau_h$ is bounded by $\sum_{u=1}^{\eta_h} (\lceil{B_{i,j}^{rd}/ T_h}\rceil+1$). Multiplying each element of this summation by $G_{h,k}+\epsilon$ therefore gives the maximum waiting time caused by the GPU requests of $\tau_h$, which is exactly used as the second term of the equation.
\end{proof}

\begin{lemma}
	\label{lemma:GPU_server_waiting_delay_jd}
	{\em [Job-driven analysis]} The maximum waiting time given by the GPU requests of other tasks during a single job execution of a task $\tau_i$ under the server-based approach is given by:	
	\begin{equation} \label{eq:GPU_server_waiting_delay_jd}
	\begin{split}
	B_{i}^{jd}&=\eta_i \cdot \max_{\pi_l<\pi_i\land 1\le k\le \eta_l}(G_{l,k}+\epsilon) 
	+ \sum_{\pi_h > \pi_i\land 1\le k \le \eta_h}\left(\left\lceil{W_i\over T_h}\right\rceil +1 \right) (G_{h,k}+\epsilon)
	\end{split}
	\end{equation}
	where $W_i$ is the response time of $\tau_i$. 
\end{lemma}
\begin{proof}
	During the execution of any job of $\tau_i$, in the worst case, every GPU request of $\tau_i$ may wait for the completion of the longest lower-priority GPU request, which is captured by the first term of the equation. In addition, due to the priority queue of the GPU server, higher-priority GPU requests made during the response time of $\tau_i$ ($W_i$) can introduce waiting time to the GPU requests of a single job of $\tau_i$. The second term of the equation upper-bounds the maximum waiting time given by higher-priority GPU requests during $W_i$, and it can be proved in the same manner as in the proof of Lemma~\ref{lemma:GPU_server_waiting_delay_rd}.
\end{proof}

%From Eqs.\eqref{eq:GPU_server_waiting_delay_rd} and \eqref{eq:GPU_server_waiting_delay_jd}, we can expect the difference in performance between the request-driven and job-driven analyses. For instance, the request-driven analysis may give a tighter bound than the job-driven analysis when tasks have a small number of GPU segments. We will empirically investigate their characteristics in Section~\ref{evaluation}.

The response time of a task $\tau_i$ is affected by the presence of the GPU server on $\tau_i$'s core. If $\tau_i$ is allocated on a different core than the GPU server, the worst-case response time of $\tau_i$ under the server-based approach is given by:
\begin{equation} \label{eq:GPU_server_task_sched_no_server}
\begin{split}
W_i^{n+1}=&C_i+B^{gpu}_i+\sum_{\tau_h\in \mathbb{P}(\tau_i)\land \pi_h>\pi_i}\Big\lceil{{W_i^n+(W_h-C_h)} \over {T_h}}\Big\rceil C_h
\end{split}
\end{equation}
where $\mathbb{P}(\tau_i)$ is the CPU core on which $\tau_i$ is allocated. The recurrence computation terminates when $W_i^{n+1} = W_i^n$, and $\tau_i$ is schedulable if $W_i^n \le D_i$. To compute $B_i^{jd}$ at the $n$-th iteration, $W_i^{n-1}$ can be used in Eq.~\eqref{eq:GPU_server_waiting_delay_jd}. It is worth noting that, as captured in the third term, the GPU segments of higher-priority tasks do {\em not} cause any direct interference to the normal execution segments of $\tau_i$ because they are executed by the GPU server that runs on a different core.

If $\tau_i$ is allocated on the same core as the GPU server, the worst-case response time of $\tau_i$ is given by:
\begin{equation} \label{eq:GPU_server_task_sched}
\begin{split}
W_i^{n+1}=&C_i+B^{gpu}_i+\sum_{\tau_h\in \mathbb{P}(\tau_i)\land \pi_h>\pi_i}\Big\lceil{{W_i^n+(W_h-C_h)} \over {T_h}}\Big\rceil C_h\\
&+\sum_{\tau_j\ne\tau_i \land \eta_j > 0} \Big\lceil{{W_i^{n}+\{D_j-(G^m_j+2\eta_j\epsilon)\}} \over{T_j}} \Big\rceil(G^m_j+2\eta_j\epsilon)
\end{split}
\end{equation}
where $G_j^m$ is the sum of the WCETs of miscellaneous operations in $\tau_j$'s GPU access segments, i.e., $G_j^m=\sum_{k=1}^{\eta_j}G_{j,k}^m$. 

Under the server-based approach, both, the GPU-using tasks, as well as the GPU server task, can self-suspend. Hence, we use the following lemma given by Bletsas et al.~\cite{Bletsas_TR15} to prove Eqs.~\eqref{eq:GPU_server_task_sched_no_server} and \eqref{eq:GPU_server_task_sched}:
\begin{lemma}
	\label{lemma:self_suspending_task}	
	{\em [from \cite{Bletsas_TR15}]}
	The worst-case response time of a self-suspending task $\tau_i$ is upper-bounded by:
	\begin{equation}\label{eq:self_suspending_task}
	\begin{split}
	W_i^{n+1}=C_i+\sum_{\tau_h\in \mathbb{P}(\tau_i)\land \pi_h>\pi_i} \Big\lceil{{W_i^n+(W_h-C_h)}\over{T_h}}\Big\rceil C_h
	\end{split}
	\end{equation}
\end{lemma}
Note that $D_h$ can be used instead of $W_h$ in the summation term of Eq.~\eqref{eq:self_suspending_task}~\cite{Chen_TR16}.

\begin{theorem}
	The worst-case response time of a task $\tau_i$ under the server-based approach is given by Eqs.~\eqref{eq:GPU_server_task_sched_no_server} and \eqref{eq:GPU_server_task_sched}.
\end{theorem}
\begin{proof}
	To account for the maximum GPU request handling time of $\tau_i$, $B_i^{gpu}$ is added in both Eqs.~\eqref{eq:GPU_server_task_sched_no_server} and \eqref{eq:GPU_server_task_sched}. In the ceiling function of the third term of both equations, $(W_h-C_h)$ accounts for the self-suspending effect of higher-priority GPU-using tasks (by Lemma~\ref{lemma:self_suspending_task}). With these, Eq.\eqref{eq:GPU_server_task_sched_no_server} upper-bounds the worst-case response time of $\tau_i$ when it is allocated on a different core than the GPU server.
	
	The main difference between Eq.~\eqref{eq:GPU_server_task_sched} and Eq.~\eqref{eq:GPU_server_task_sched_no_server} is the last term, which captures the worst-case interference from the GPU server task. The execution time of the GPU server task is bounded by summing up the worst-case miscellaneous operations and the server overhead caused by GPU requests from all other tasks ($G^m_j+2\eta_j\epsilon$). Since the GPU server self-suspends during CPU-inactive time intervals, adding $\{D_j\!-\!(G^m_j+2\eta_j\epsilon)\}$ to $W_i^n$ in the ceiling function captures the worst-case self-suspending effect (by Lemma~\ref{lemma:self_suspending_task}). These factors are exactly captured by the last term of Eq.~\eqref{eq:GPU_server_task_sched}. Hence, it upper-bounds task response time in the presence of the GPU server. 	
\end{proof}

\subsection{Task Allocation with the GPU Server}

As shown in Eqs.~\eqref{eq:GPU_server_task_sched} and \eqref{eq:GPU_server_task_sched_no_server}, the response time of a task $\tau_i$ with the server-based approach is affected by the presence of the GPU server on $\tau_i$'s core. This implies that, in order to achieve better schedulability, the GPU server should be considered when allocating tasks to cores. 

Under partitioned scheduling, finding the optimal task allocation can be modeled as the bin-packing problem, which is known to be NP-complete~\cite{Johnson_74}. Hence, bin-packing heuristics, such as first-fit decreasing and worst-fit decreasing, have been widely used as practical solutions to the task allocation problem. Such heuristics first sort tasks in decreasing order of utilization and then allocate them to cores. In case of the GPU server, its utilization depends on the frequency and length of miscellaneous operations of GPU access segments. Specifically, the utilization of the GPU server is given by:
\begin{equation} \label{eq:GPU_server_util}
\begin{split}
U_{server}=\sum_{\forall \tau_i:\eta_i>0}{{G_i^m+2\eta_i\epsilon}\over{T_i}}
\end{split}
\end{equation}
%Note that the GPU server incurs at most two times of overhead per GPU request (by Lemma~\ref{lemma:server_overhead}).
With this utilization, the GPU server can be sorted and allocated together with other regular tasks using conventional bin-packing heuristics. We will use this approach for schedulability experiments in Section~\ref{schedulability_experiments}.

\section{Evaluation}
\label{evaluation}

This section provides our experimental evaluation of the two different approaches for GPU access control. We first present details about our implementation and describe case study results on a real embedded platform. Next, we explore the impact of these approaches on task schedulability with randomly-generated tasksets, by using parameters based on the practical overheads measured from our implementation.
%Our focus here is to explore the impact of those approaches on task schedulability. We first generate random tasksets and capture the percentage of schedulable tasksets as the metric. Then, we conduct a case study of those two approaches on a real embedded platform. 

\subsection{Implementation}
\label{implementation}

%\smallskip
%\noindent\textbf{Implementation.}
We implemented prototypes of the synchronization-based and the server-based approaches on a SABRE Lite board~\cite{SabreLite}. The board is equipped with an NXP i.MX6 Quad SoC that has four ARM Cortex-A9 cores and one Vivante GC2000 GPU. We ran an NXP Embedded Linux kernel version 3.14.52 patched with Linux/RK~\cite{LinuxRK} version 1.6\footnote{Linux/RK is available at \url{http://rtml.ece.cmu.edu/redmine/projects/rk/}.}, and used the Vivante v5.0.11p7.4 GPU driver along with OpenCL 1.1 (Embedded Profile) for general-purpose GPU programming. We also configured each core to run at its maximum frequency, 1~GHz.

Linux/RK provides a kernel-level implementation of MPCP which we used to implement the synchronization-based approach. Under our implementation, each GPU-using task first acquires an MPCP-based lock, issues memory copy and GPU kernel execution requests in an asynchronous manner, and uses OpenCL events to busy-wait on the CPU till the GPU operation completes, before finally releasing the lock. 

To implement the server-based approach, we set up shared memory regions between the server task and each GPU-using task, which are used to share GPU input/output data. POSIX signals are used by the GPU-using tasks and the server to notify GPU requests and completions, respectively. The server has an initialization phase, during which, it initializes shared memory regions and obtains GPU kernels from the GPU binaries (or source code) of each task. Subsequently, the server uses these GPU kernels whenever the corresponding task issues a GPU request. As the GPU driver allows suspensions during GPU requests, the server task issues memory copy and GPU kernel execution requests in an asynchronous manner, and suspends by calling the \texttt{clFinish()} API function provided by OpenCL.

\smallskip\noindent\textbf{GPU Driver and Task-Specific Threads.} The OpenCL implementation on the i.MX6 platform spawns user-level threads in order to handle GPU requests and to notify completions. Under the synchronization-based approach, OpenCL spawns multiple such threads for each GPU-using task, whereas under the server-based approach, such threads are spawned only for the server task. To eliminate possible scheduling interference, the GPU driver process, as well as the spawned OpenCL threads, are configured to run at the highest real-time priority in all our experiments.

\subsection{Practical Evaluation}
\label{case_study}
\smallskip\noindent\textbf{Overheads.}
We measured the practical worst-case overheads for both the approaches in order to perform schedulability analysis. Each worst-case overhead measurement involved examining 100,000 readings of the respective operations measured on the i.MX6 platform. \figref{mpcp_overhead} shows the mean and 99.9th percentile of the MPCP lock operations and \figref{server_overhead} shows the same for server related overheads.

Under the synchronization-based approach with MPCP, overhead occurs while acquiring and releasing the GPU lock. Under the server-based approach, overheads involve waking up the server task, performing priority queue operations (i.e., server execution delay), and notifying completion to wake up the GPU-using task after finishing GPU computation.
 
To safely take into consideration the worst-case overheads, we use the 99.9th percentile measurements for each source of delay in our experiments. This amounts to a total of 14.0~$\mu$s lock-related delay under the synchronization-based approach, and a total of 44.97~$\mu$s delay for the server task under the server-based approach.

%Under the server-based approach, the worst-case observed overhead,~$\epsilon$, was 44.97~$\mu$s: 17.3~$\mu$s for waking up the GPU server (from the time a GPU-using task sends a request to the time the server receives it), 6.67~$\mu$s for checking the request queue, 51.0~$\mu$s for initiating GPU computation, and 21.0~$\mu$s for notifying completion and waking up a GPU-using task. 

\begin{figure}[t]
	\centering
	\subfloat{
		\includegraphics[width=0.55\columnwidth]{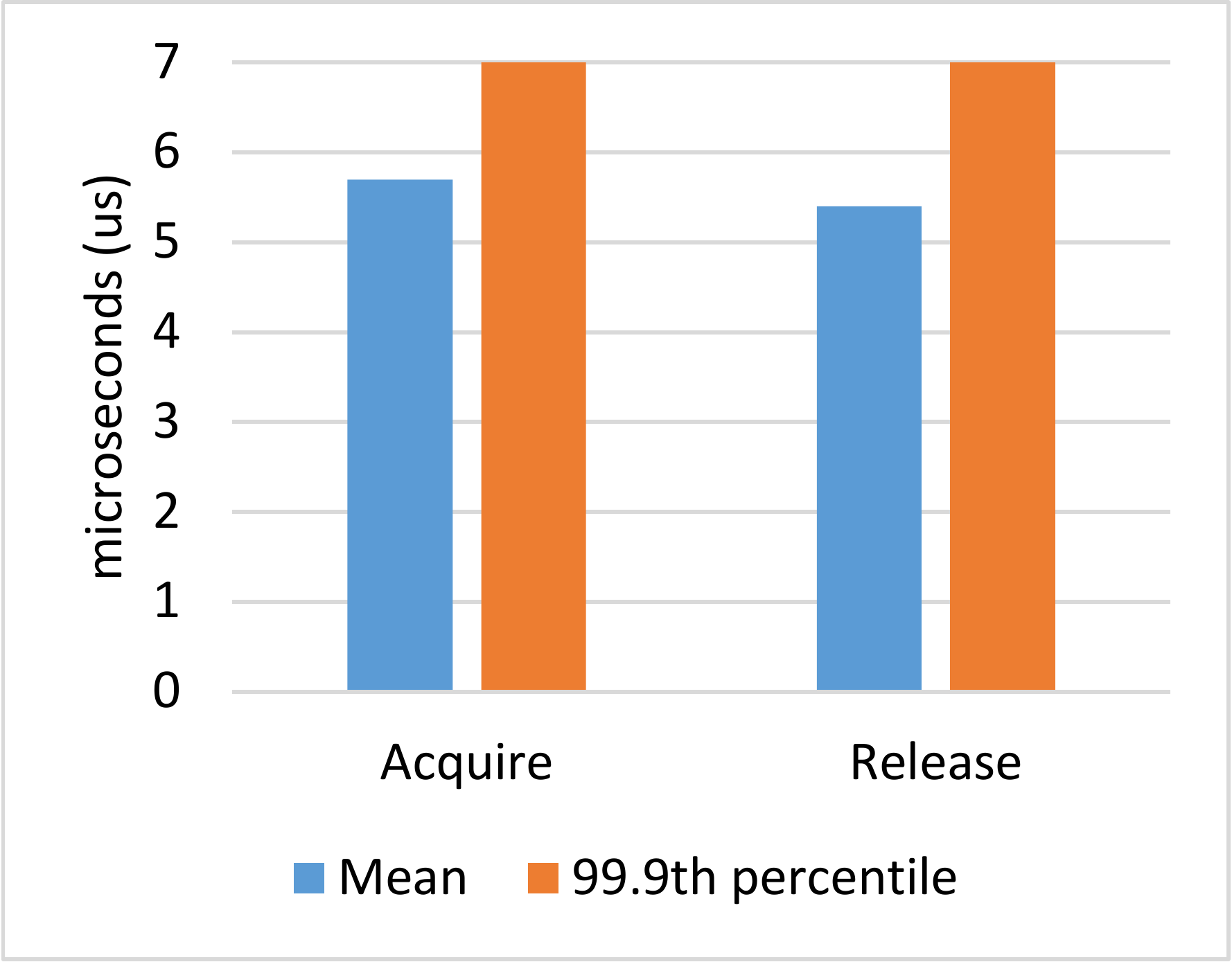}
	}\\
	\caption{MPCP lock overhead}
	\label{fig:mpcp_overhead}
\end{figure}

\begin{figure}[t]
	\centering
	\subfloat{
		\includegraphics[width=0.7\columnwidth]{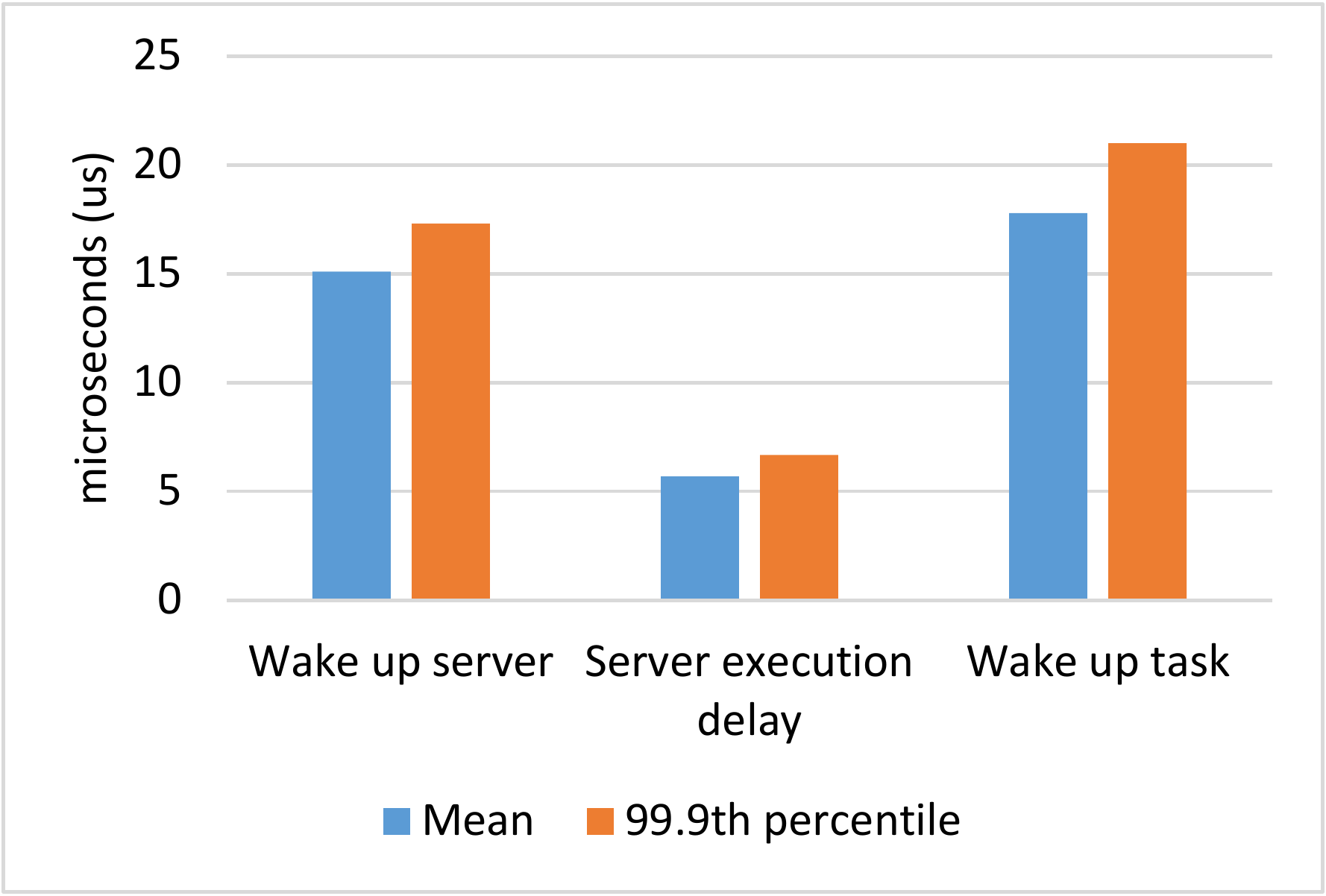}
	}
	\caption{Server task overheads}
	\label{fig:server_overhead}
\end{figure}

\begin{table*}[t]
\centering
\footnotesize
\caption{Tasks used in the case study}
\label{tab:case_study_tasks}
\begin{tabular}{|c|c|c|c|c|c|c|c|}
\hline
\textbf{Task $\tau_i$} & \textbf{Task name} & \textbf{$C_i$ (in ms)} & \textbf{$\boldsymbol{\eta_i}$} & \textbf{$G_i$ (in ms)} & \textbf{$T_i = D_i$ (in ms)} & \textbf{Core} & \textbf{Priority} \\ \hline
$\tau_1$ & \texttt{workzone}      & 20                 & 2            & $G_{1,1}$ = 95, $G_{1,2}$ = 47   & 300                    & 0             & 70                \\ \hline
$\tau_2$ & \texttt{cpu\_matmul1}  & 215                & 0            & 0                  & 750                    & 0             & 67                \\ \hline
$\tau_3$ & \texttt{cpu\_matmul2}  & 102                & 0            & 0                  & 300                    & 1             & 69                \\ \hline
$\tau_4$ & \texttt{gpu\_matmul1}  & 0.15               & 1            & $G_{4,1}$ = 19                 & 600                    & 1             & 68                \\ \hline
$\tau_5$ & \texttt{gpu\_matmul2}  & 0.15               & 1            & $G_{5,1}$ = 38                 & 1000                   & 1             & 66                \\ \hline
\end{tabular}
\end{table*}

\begin{figure*}[t]
	\centering
	\subfloat[Synchronization-based Approach (MPCP)]{\label{fig:case_study_mpcp}
		\includegraphics[width=1\textwidth]{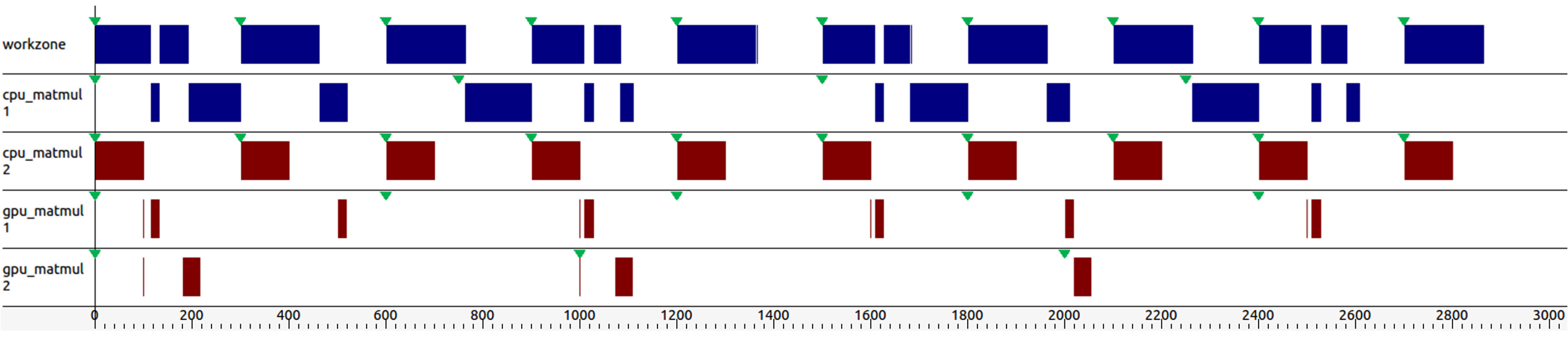}
	}\\
	\subfloat[Server-based Approach]{\label{fig:case_study_server}
		\includegraphics[width=1\textwidth]{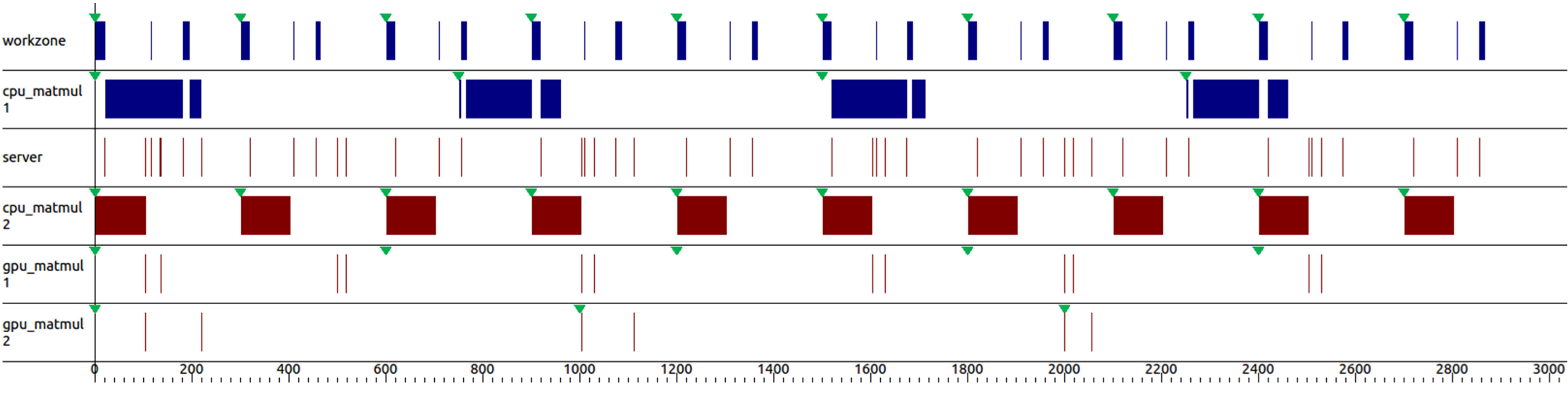}
	}
	\caption{Task execution timeline during one hyperperiod (3,000~ms)}
	\label{fig:case_study_result}
\end{figure*}

\smallskip\noindent\textbf{Case Study.}
We present a case study motivated by the software system of the self-driving car developed at CMU~\cite{Wei_IV13}. Among various algorithmic tasks of the car, we chose a GPU accelerated version of the workzone recognition algorithm~\cite{Lee_ITSC13} (\texttt{workzone}) that periodically processes images collected from a camera. Two GPU-based matrix multiplication tasks (\texttt{gpu\_matmul1} and \texttt{gpu\_matmul2}), and two CPU-bound tasks (\texttt{cpu\_matmul1} and \texttt{cpu\_matmul2}) are also used to represent a subset of other tasks of the car. Unique task priorities are assigned based on the Rate-Monotonic policy~\cite{Liu_Layland} and each task is pinned to a specific core as described in \tableref{case_study_tasks}. Of these, the \texttt{workzone} task has two GPU segments per job whereas both the GPU-based matrix multiplication tasks have a single GPU segment. Under the server-based approach, the \texttt{server} task is pinned to CPU Core 1 and is run with real-time priority 80. In order to avoid unnecessary interference while recording traces, the task-specific OpenCL threads are pinned on a separate CPU core for both approaches. All tasks are released at the same time using Linux/RK APIs. CPU scheduling is performed using the \texttt{SCHED\_FIFO} policy, and CPU execution traces are collected for one hyperperiod (3,000 ms) as shown in~\figref{case_study_result}. 

The CPU-execution traces for the synchronization and server-based approaches are shown in \figsubref{case_study_mpcp} and \figsubref{case_study_server}, respectively. Tasks executing on Core 0 are shown in blue whereas tasks executing on Core 1 are shown in red. It is immediately clear that the server-based approach allows suspension of tasks while they are waiting for the GPU request to complete. In particular, we make the following key observations from the case study:

\begin{enumerate}
	\item Under the synchronization-based approach, tasks suspend when they wait for the GPU lock to be released, but they do not suspend while using the GPU. On the contrary, under the server-based approach, tasks suspend even when their GPU segments are being executed. The results show that our proposed server-based approach can be successfully implemented and used on a real platform.
	\item The response time of \texttt{cpu\_matmul1} under the synchronization-based approach is significantly larger than that under the server-based approach, i.e., 520.68~ms vs. 219.09~ms in the worst case, because of the busy-waiting problem discussed in Section~\ref{GPU_limitations_of_synchronization_approach}. 
\end{enumerate}

\subsection{Schedulability Experiments}
\label{schedulability_experiments}

\begin{table}[t]
	\centering
	
	{
		\footnotesize
		\caption[Base parameters for GPGPU access control experiments]{Base parameters for taskset generation}\label{tab:GPU_taskset_param}
		\begin{tabular}{l|c}
			\hline
			\textbf{Parameters} & \textbf{Values}\\\hline
			Number of CPU cores ($N_P$) & 4, 8\\
			Number of tasks ($n$)& [$2 N_P$, $5 N_P$]\\
			Task utilization ($U_i$) & [0.05, 0.2] \\
			Task period and deadline ($T_i=D_i$) & [30, 500] ms\\
			Percentage of GPU-using tasks & [10, 30] \%\\
%			Taskset utilization per core & [30, 50] \% \\
			Ratio of GPU segment len. to normal WCET ($G_i/C_i$)	& [10, 30] \%\\
			Number of GPU segments per task ($\eta_i$) & [1, 3]\\			
			Ratio of misc. operations in $G_{i,j}$ ($G_{i,j}^m/G_{i,j}$)& [10, 20] \%\\
			GPU server overhead ($\epsilon$) & 50 $\mu$s\\
			\hline
		\end{tabular}
	}
\end{table}

%\smallskip
\noindent\textbf{Taskset Generation.}
We used 10,000 randomly-generated tasksets for each experimental setting. Unless otherwise mentioned, the base parameters given in \tableref{GPU_taskset_param} are used for taskset generation. The GPU-related parameters are inspired from the observations from our case study and the GPU workloads used in prior work~\cite{Kato_ATC11, Kato_ATC12}. Other task parameters are similar to those used in the literature.
Systems with four and eight CPU cores ($N_P=\{4, 8\}$) are considered. 
%However, due to space limit, we only present results with $N_P=4$ in this paper (see Appendix in \cite{AnonymousTR} for results with $N_P=8$).
%The ranges of the length of a GPU access segment and the WCET of miscellaneous GPU operations given in \tableref{GPU_taskset_param} reflect the measurements of GPU workloads with reasonable data size reported in prior work~\cite{Kato_ATC11, Kato_ATC12}.
For each taskset, $n$ tasks are generated where $n$ is uniformly distributed over $[2N_P, 5N_P]$. 
%To generate each taskset, the number of CPU cores in the system and the number of tasks for each core are first created based on the parameters in \tableref{GPU_taskset_param}. 
A subset of the generated tasks is chosen at random, corresponding to the specified percentage of GPU-using tasks, to include GPU segments. Task period $T_i$ and utilization $U_i$ are randomly selected from the defined minimum and maximum period ranges. Task deadline $D_i$ is set equal to $T_i$. 
%On each core, the taskset utilization is split into $k$ random-sized pieces, where $k$ is the number of tasks per core. The size of each piece represents the utilization $U_i$ of the corresponding task $\tau_i$, i.e., $U_i=(C_i+G_i)/T_i$. 
Recall that the utilization of a task $\tau_i$ is defined as $U_i=(C_i+G_i)/T_i$. If $\tau_i$ is a CPU-only task, $C_i$ is set to $U_i\cdot T_i$ and $G_i$ is set to zero. If $\tau_i$ is a GPU-using task, the given ratio of the accumulated GPU segment length to the WCET of normal segments is used to determine the values of $C_i$ and $G_i$. $G_i$ is then split into $\eta_i$ random-sized pieces, where $\eta_i$ is the number of $\tau_i$'s GPU segments chosen randomly from the specified range. For each GPU segment $G_{i,j}$, the values of $G_{i,j}^e$ and $G_{i,j}^m$ are determined by the ratio of miscellaneous operations given in \tableref{GPU_taskset_param}, assuming $G_{i,j}=G_{i,j}^e+G_{i,j}^m$.
Finally, task priorities are assigned by the Rate-Monotonic policy~\cite{Liu_Layland}, with arbitrary tie-breaking.

\smallskip\noindent\textbf{Results.}
We captured the percentage of schedulable tasksets where all tasks met their deadlines. For the synchronization-based approach, two multiprocessor locking protocols are used: \texttt{MPCP} and \texttt{FMLP+}. Schedulability under MPCP and FMLP+ is tested using the analysis developed by Lakshmanan et al.~\cite{Lakshmanan_RTSS09} and Brandenburg~\cite{Brandenburg_11}\footnote{We used the FMLP+ analysis for preemptive partitioned fixed-priority scheduling given in Section 6.4.3 of \cite{Brandenburg_11}.}, respectively. Both analyses are appropriately modified based on the correction by Chen et al.~\cite{Chen_TR16}. 
%MPCP~\cite{MPCP} is used with the task schedulability test given in the Appendix.
We considered both zero and non-zero locking overhead under the synchronization-based approach, but there was no appreciable difference between them. Hence, only the results with zero overhead are presented in the paper. 
Tasks are allocated to cores by using the worst-fit decreasing (WFD) heuristic in order to balance the load across cores. Under the server-based approach, the GPU server is allocated along with other tasks by WFD. Eqs.~\eqref{eq:GPU_server_task_sched_no_server} and \eqref{eq:GPU_server_task_sched} are used for task schedulability tests. We set the GPU server overhead~$\epsilon$ to 50~$\mu$s, which is slightly larger than the measured overheads from our implementation presented in Section~\ref{case_study}. Two versions of the GPU server are considered to show the analytical improvement of this work: \texttt{Server-RD}, which is our initial work~\cite{Kim_RTCSA17} using only the request-driven approach, and \texttt{Server-RD+JD}, which is the one newly proposed in this paper. Unless otherwise mentioned, we will limit our focus to Server-RD+JD and call it as the server-based approach.

\begin{figure}[t]
	\centering
	\subfloat[$N_P=4$]{\label{fig:expr_gpu_segment_length_4cores}
		\includegraphics[width=0.48\columnwidth]{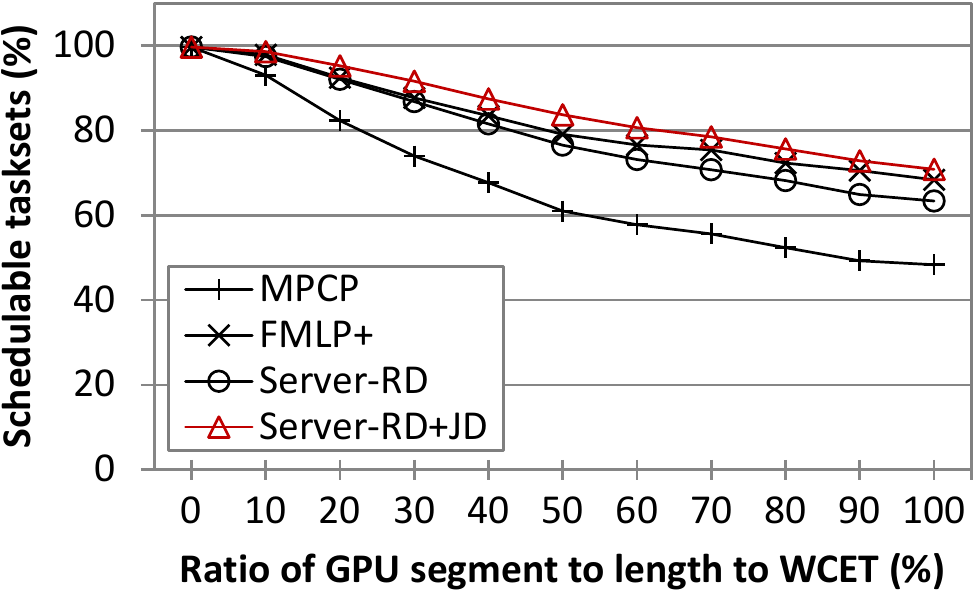}
	}
	\subfloat[$N_P=8$]{\label{fig:expr_gpu_segment_length_8cores}
		\includegraphics[width=0.48\columnwidth]{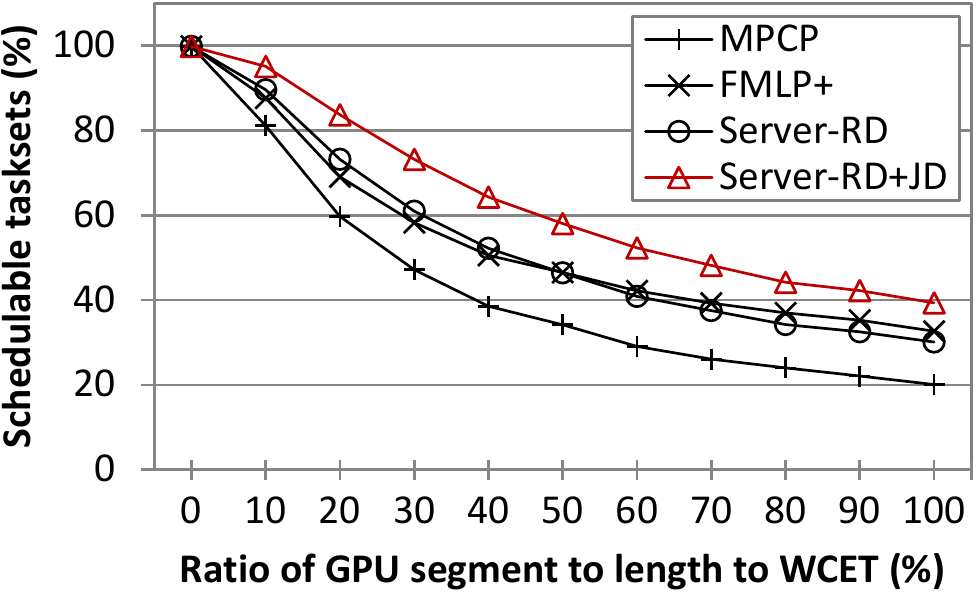}
	}
	\caption{Schedulability w.r.t. the GPU segment length}
	\label{fig:expr_gpu_segment_length}
\end{figure}

\figref{expr_gpu_segment_length} shows the percentage of schedulable tasksets as the ratio of the accumulated GPU segment length ($G_i$) increases. %The solid lines denote the results with $N_P=4$ and the dotted lines denote results with $N_P=8$.
In general, the percentage of schedulable tasksets is higher when $N_P=4$, compared to when $N_P=8$. This is because the GPU is contended for by more tasks as the number of cores increases.
The server-based approach outperforms both of the synchronization-based approaches (MPCP and FMLP+) in all cases of \figref{expr_gpu_segment_length}. This is mainly due to the fact that the server-based approach allows other tasks to use the CPU while the GPU is being used. While MPCP generally performs worse than FMLP+ in our experiments, we suspect that this is due to the pessimism in the analysis~\cite{Lakshmanan_RTSS09} that computes an upper bound by the sum of the maximum per-request delay, similarly to the request-driven analysis shown in Eq.~\ref{eq:GPU_server_waiting_delay_rd}. If it is extended with the job-driven analysis like Eq.~\ref{eq:GPU_server_waiting_delay_jd}, its performance will likely be improved.

\begin{figure}[t]
	\centering
	\subfloat[$N_P=4$]{\label{fig:expr_gpu_using_tasks_4cores}
		\includegraphics[width=0.48\columnwidth]{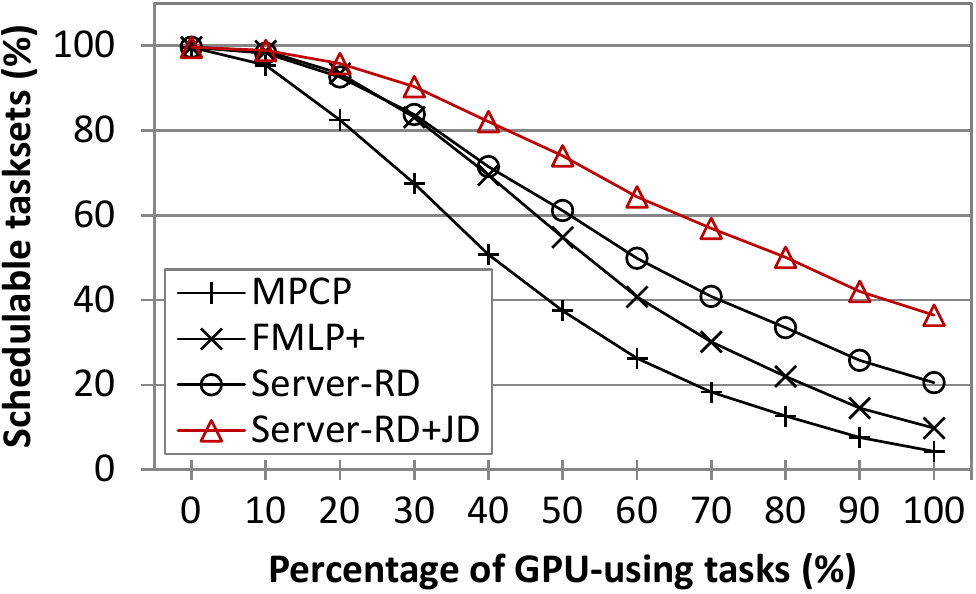}
	}
	\subfloat[$N_P=8$]{\label{fig:expr_gpu_using_tasks_8cores}
		\includegraphics[width=0.48\columnwidth]{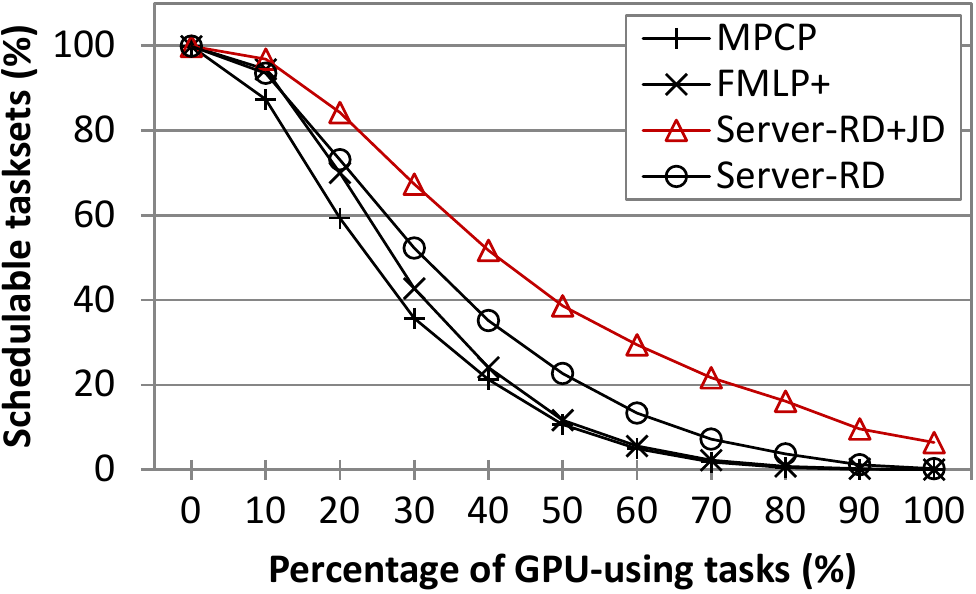}
	}
	\caption{Schedulability w.r.t. the percentage of GPU-using tasks}
	\label{fig:expr_gpu_using_tasks}
\end{figure}

\figref{expr_gpu_using_tasks} shows the percentage of schedulable tasksets as the percentage of GPU-using tasks increases. The left-most point on the x-axis represents that all tasks are CPU-only tasks, and the right-most point represents that all tasks access the GPU. Under all approaches, the percentage of schedulable tasksets reduces as the percentage of GPU-using tasks increases. However, the server-based approach significantly outperforms the other approaches, with as much as 38\% and 27\% more tasksets being schedulable than MPCP and FMLP+, respectively, when the percentage of GPU-using tasks is 70\% and $N_P=4$.

\begin{figure}[t]
	\centering
	\subfloat[$N_P=4$]{\label{fig:expr_nr_tasks_4cores}
		\includegraphics[width=0.48\columnwidth]{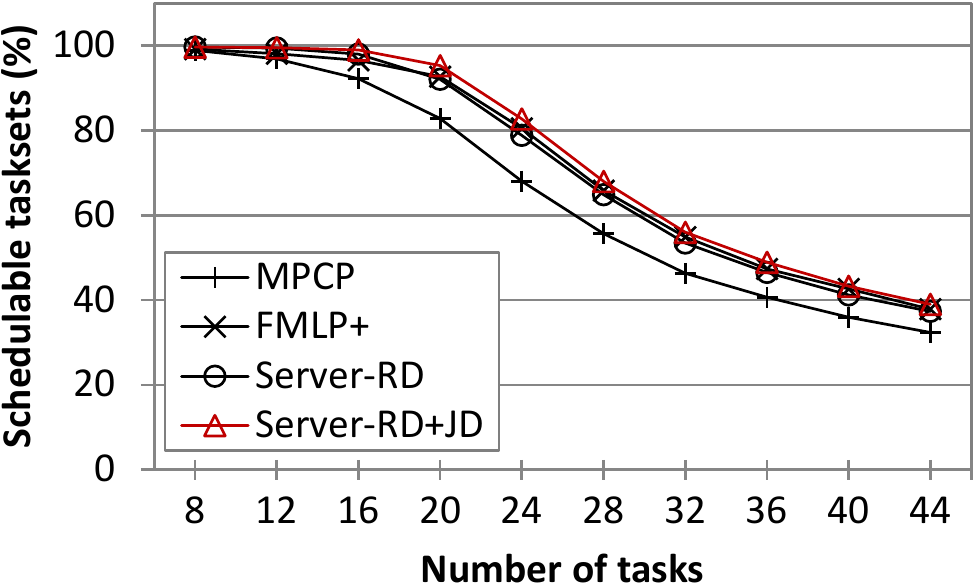}
	}
	\subfloat[$N_P=8$]{\label{fig:expr_nr_tasks_8cores}
		\includegraphics[width=0.48\columnwidth]{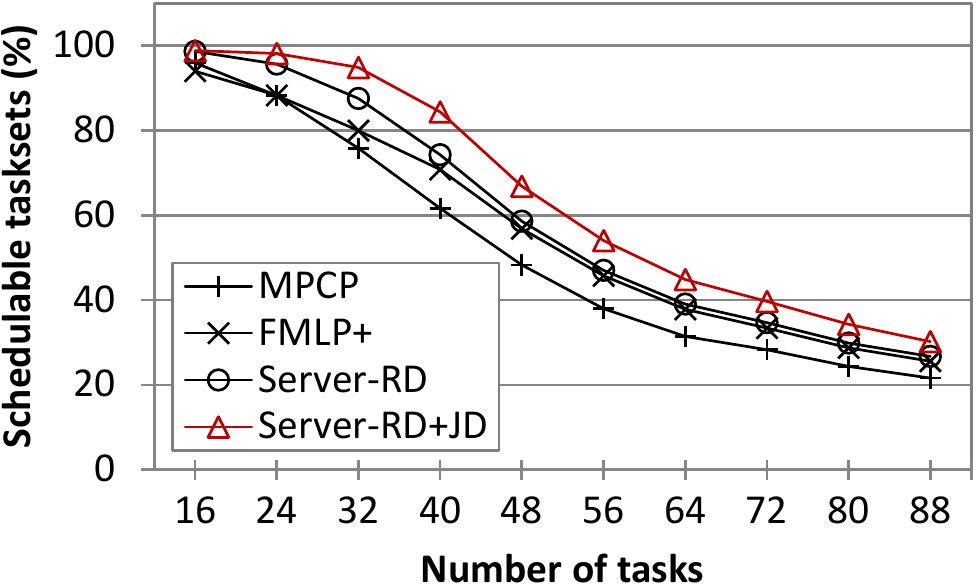}
	}
	\caption{Schedulability w.r.t. the number of tasks}
	\label{fig:expr_nr_tasks}
\end{figure}

\begin{figure}[t]
	\centering
	\subfloat[$N_P=4$]{\label{fig:expr_nr_gpu_segments_4cores}
		\includegraphics[width=0.48\columnwidth]{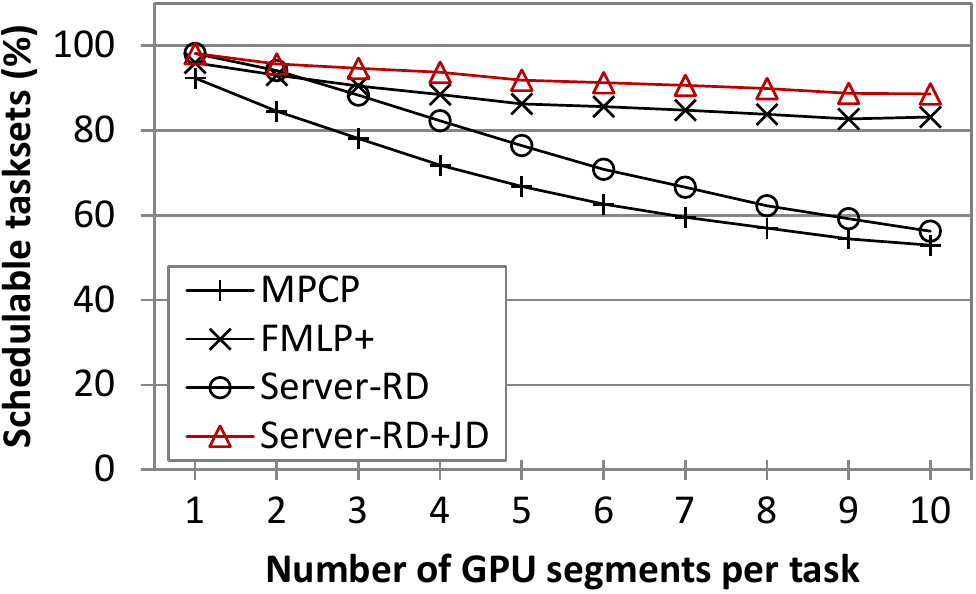}
	}
	\subfloat[$N_P=8$]{\label{fig:expr_nr_gpu_segments_8cores}
		\includegraphics[width=0.48\columnwidth]{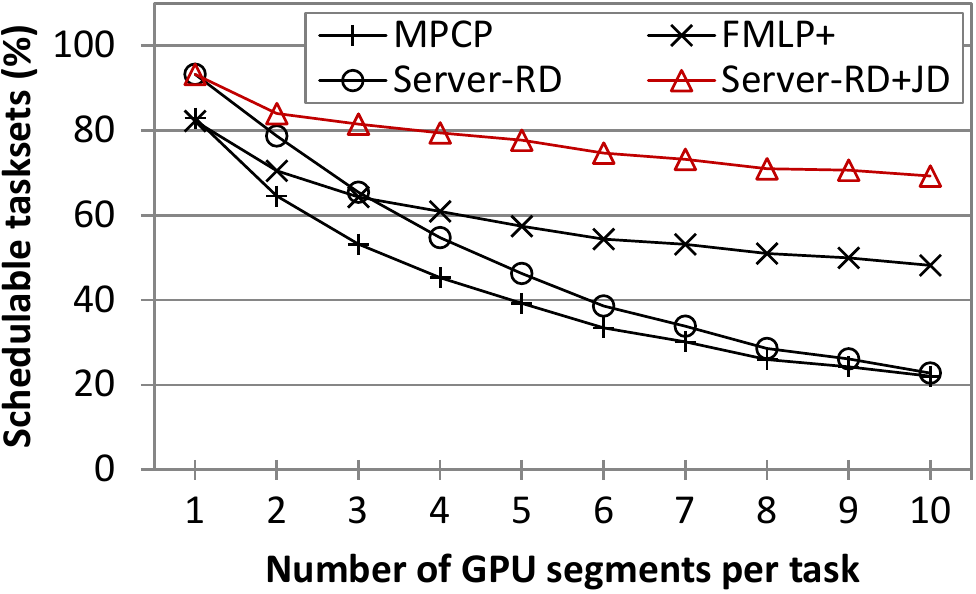}
	}
	\caption{Schedulability w.r.t. the number of GPU segments}
	\label{fig:expr_nr_gpu_segments}
\end{figure}

The benefit of the server-based approach is also observed with changes in other task parameters. In~\figref{expr_nr_tasks}, the percentage of schedulable tasksets is illustrated as the number of tasks increases. 
While the difference in schedulability between the server and FMLP+ is rather small when $N_P=4$, it becomes larger when $N_P=8$. 
This happens because the total amount of the CPU-inactive time in GPU segments increases as more tasks are considered. A similar observation is also seen in Figure~\ref{fig:expr_nr_gpu_segments}, where the number of GPU access segments per task varies. Here, the different is much noticeable as more GPU segments create larger CPU-inactive time, which is favorable to the server-based approach.

\begin{figure}[t]
	\centering
	\subfloat[$N_P=4$]{\label{fig:expr_bimodal_4cores}
		\includegraphics[width=0.48\columnwidth]{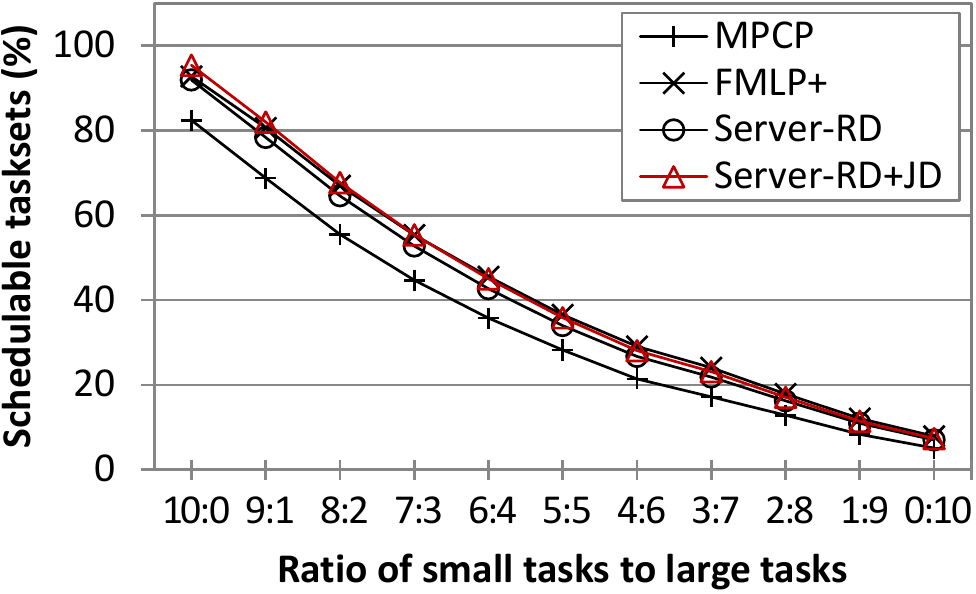}
	}
	\subfloat[$N_P=8$]{\label{fig:expr_bimodal_8cores}
		\includegraphics[width=0.48\columnwidth]{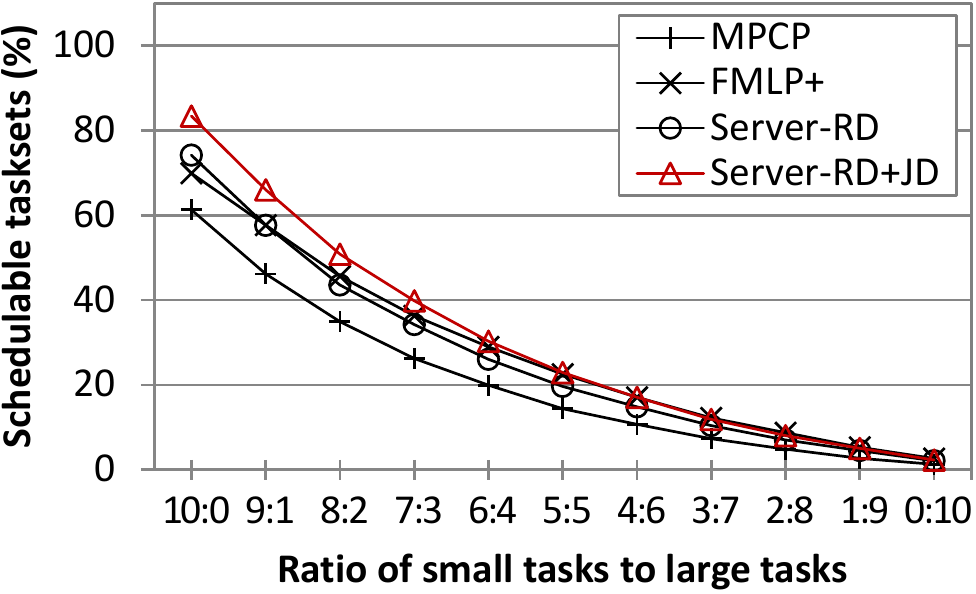}
	}
	\caption{Schedulability w.r.t. the ratio of small tasks to large tasks}
	\label{fig:expr_bimodal}
\end{figure}

The base task utilization given in \tableref{GPU_taskset_param} follows a uniform distribution. In Figure~\ref{fig:expr_bimodal}, we evaluate with bimodal distributions where the utilization of small tasks is chosen from $[0.05,0.2]$ and that of large tasks is chosen from $[0.2,0.5]$. The x-axis of each sub-figure in Figure~\ref{fig:expr_bimodal} shows the ratio of small tasks to large tasks in each taskset. While the server-based approach has higher schedulability for fewer large tasks, the gap gets smaller and the percentage of schedulable tasksets under all approaches goes down to zero as the ratio of large tasks increases, due to the high resulting utilization of generated tasksets.

\begin{figure}[t]
	\centering
	\subfloat[$N_P=4$]{\label{fig:expr_server_overhead_4cores}
		\includegraphics[width=0.48\columnwidth]{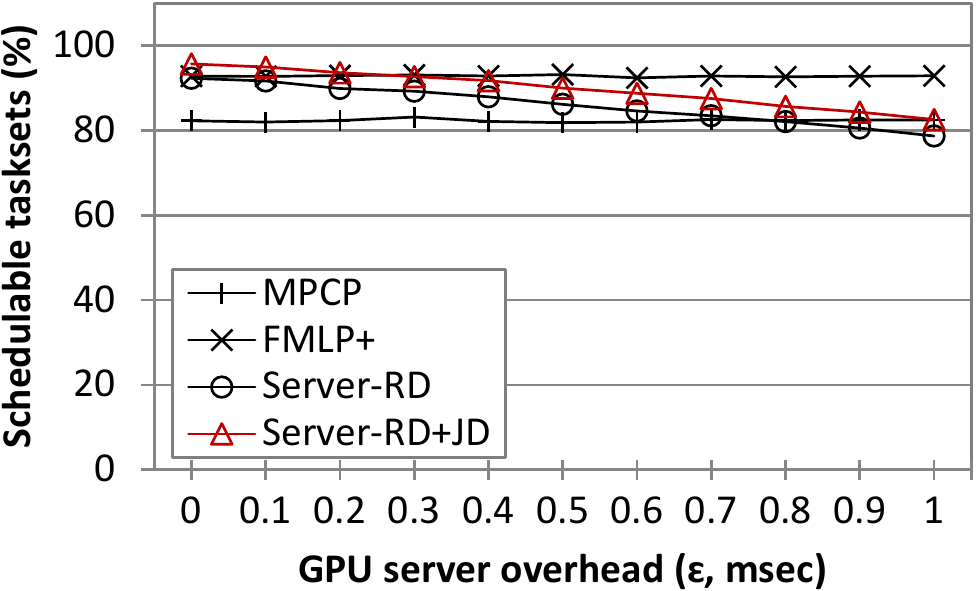}
	}
	\subfloat[$N_P=8$]{\label{fig:expr_server_overhead_8cores}
		\includegraphics[width=0.48\columnwidth]{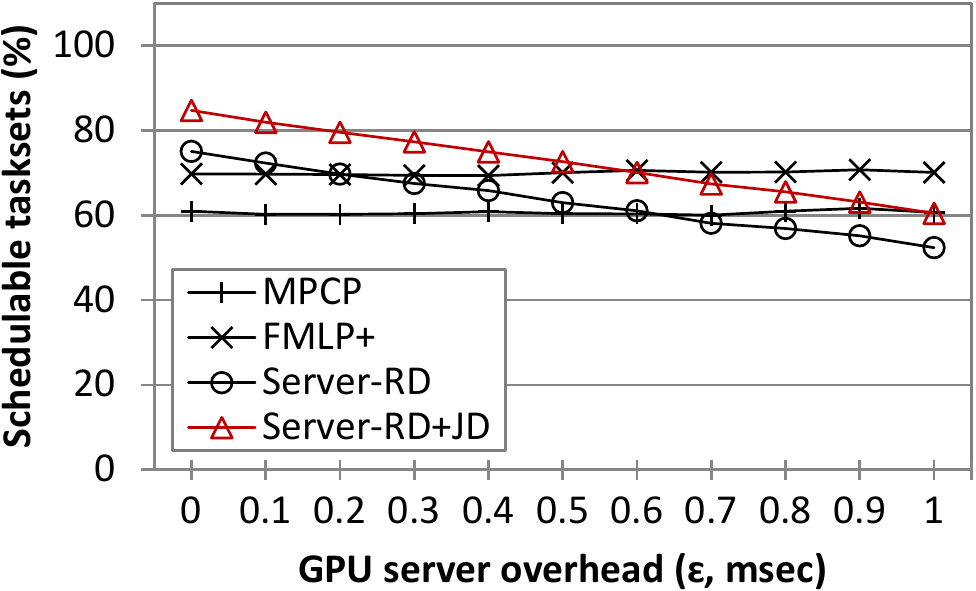}
	}
	\caption{Schedulability w.r.t. the GPU server overhead ($\epsilon$)}
	\label{fig:expr_server_overhead}
\end{figure}

We next investigate the factors that negatively impact the performance of the server-based approach. The GPU server overhead $\epsilon$ is obviously one such factor. Although an $\epsilon$ of 50~$\mu$s that we used in prior experiments is sufficient enough to upper-bound the GPU server overhead in most practical systems, we further investigate with larger $\epsilon$ values.  \figref{expr_server_overhead} shows the percentage of schedulable tasksets as the GPU server overhead~$\epsilon$ increases. Since $\epsilon$ exists only under the server-based approach, the performance of MPCP and FMLP+ is unaffected by this factor.\footnote{Only a small fluctuation in schedulability of synchronization-based approaches ($<1\%$) exists due to the nature of random parameters.} On the other hand, the performance of the server-based approach deteriorates as the overhead increases.

 \begin{figure}[t]
 	\centering
 	\subfloat[$N_P=4$]{\label{fig:expr_misc_operations_4cores}
 		\includegraphics[width=0.48\columnwidth]{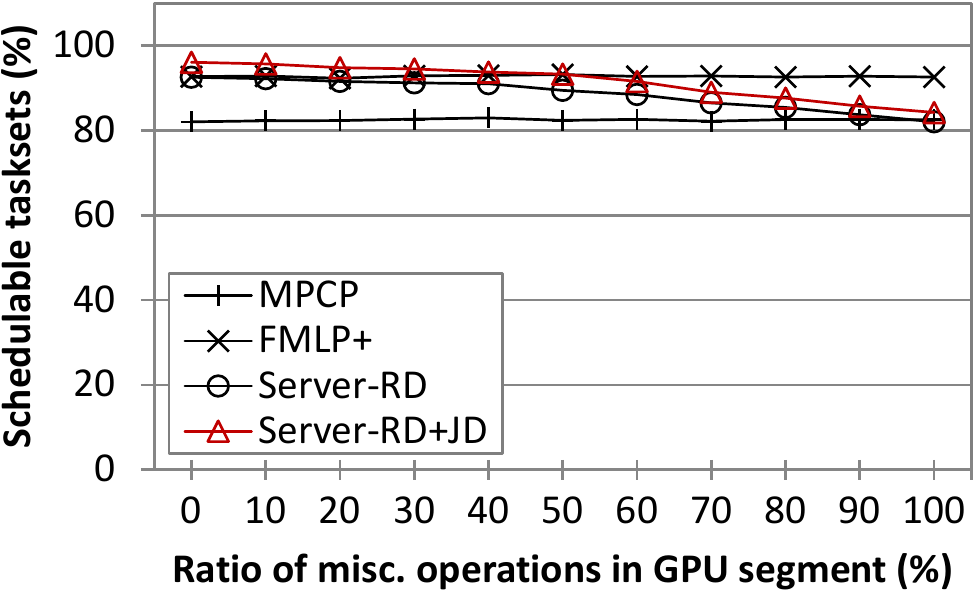}
 	}
 	\subfloat[$N_P=8$]{\label{fig:expr_misc_operations_8cores}
 		\includegraphics[width=0.48\columnwidth]{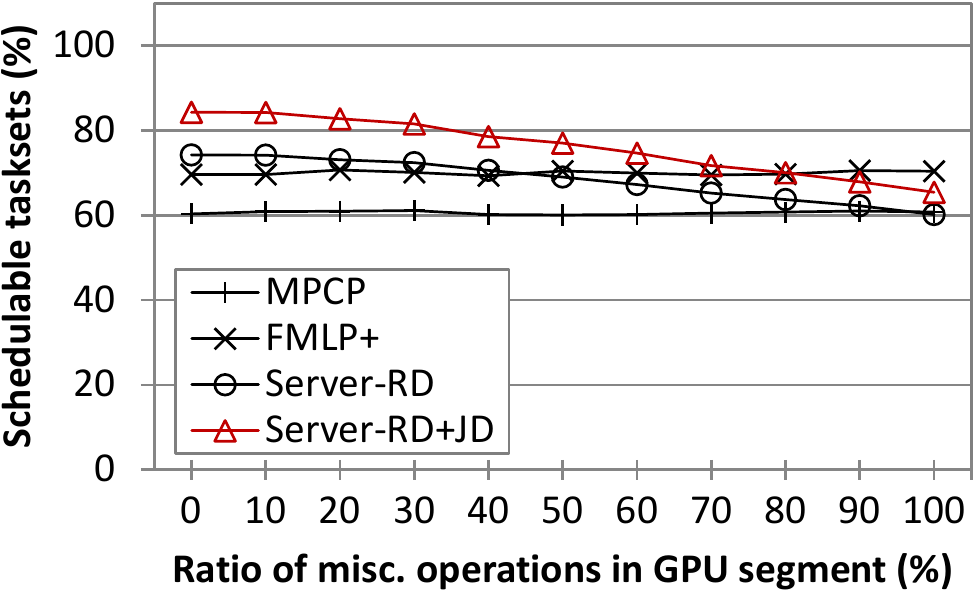}
 	}
 	\caption{Schedulability w.r.t. the ratio of miscellaneous operations in GPU access segments}
 	\label{fig:expr_misc_operations}
 \end{figure}

The length of miscellaneous operations in GPU access segments is another factor degrading the performance of the server-based approach because miscellaneous operations require the GPU server to consume a longer CPU time. \figref{expr_misc_operations} shows the percentage of schedulable tasksets as the ratio of miscellaneous operations in GPU access segments increases. As the analyses of MPCP and FMLP+ require tasks to busy-wait during their entire GPU access, their performance remains unaffected. On the other hand, as expected, the performance of the server-based approach degrades as the ratio of miscellaneous operations increases. Specifically, starting from the ratio of 60\% at $N_P=4$ and 90\% at $N_P=8$, the server-based approach underperforms FMLP+. However, we expect that such a high ratio of miscellaneous operations in GPU segments is hardly observable in practical GPU applications because memory copy is typically done by DMA and GPU kernel execution takes the majority time of GPU segments. 

 \begin{figure}[t]
	\centering
	\subfloat[$N_P=4$]{\label{fig:expr_min_task_period_4cores}
		\includegraphics[width=0.48\columnwidth]{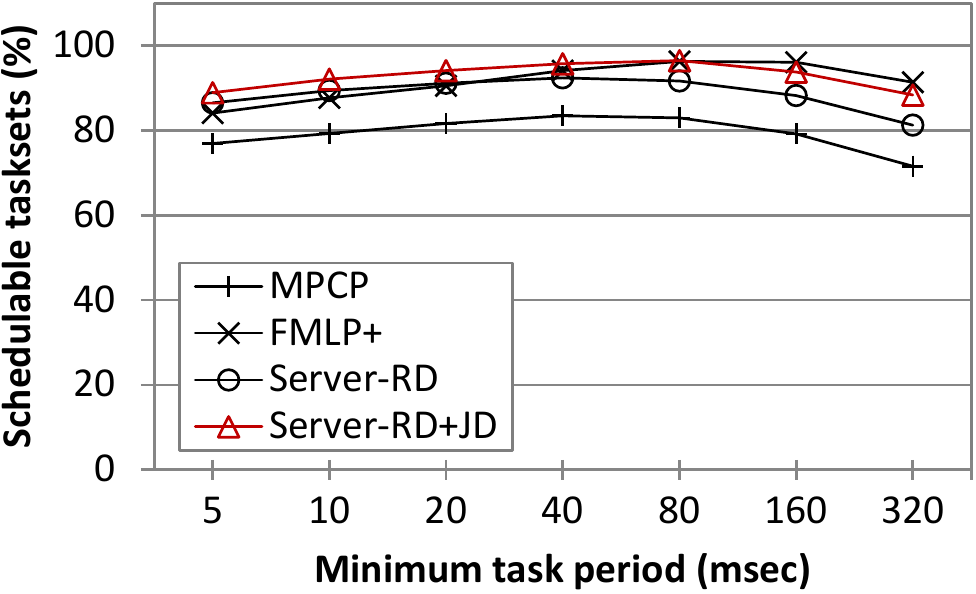}
	}
	\subfloat[$N_P=8$]{\label{fig:expr_min_task_period_8cores}
		\includegraphics[width=0.48\columnwidth]{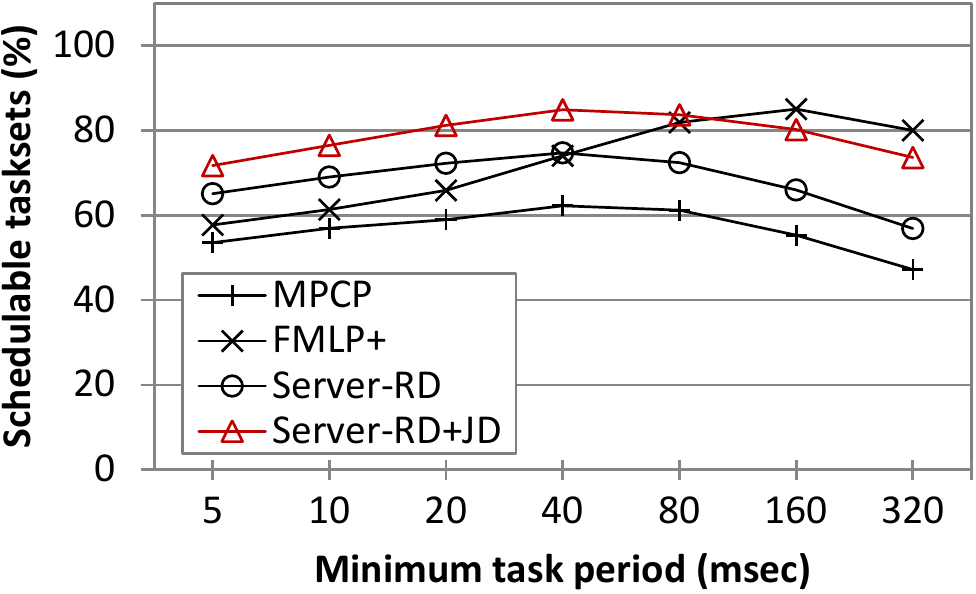}
	}
	\caption{Schedulability w.r.t. the minimum task period}
	\label{fig:expr_min_task_period}
\end{figure}

Lastly, we evaluate the impact of task periods in schedulability. Recall that task periods are uniformly chosen from $[T_{min}, T_{max}]$, where $T_{min}=20$ ms and $T_{max}=500$ ms in our base parameters. In this experiment, we fix $T_{max}$ to 500 ms and vary the value of $T_{min}$. The results are shown in Figure~\ref{fig:expr_min_task_period}. Interestingly, the server-based approach outperforms FMLP+ until $T_{min}$ reaches 80ms at $N_P=4$ and 160 ms at $N_P=8$, but it starts to underperform afterwards. We suspect that this is mainly due to the rate-monotonic (RM) priority assignment used in our experiments and the request ordering policies of the GPU server and FMLP+. The priority-based ordering of the GPU server favors higher-priority tasks and the FIFO ordering of FMLP+ ensures fairness in waiting time. If there is a large difference in task periods, it is better to give shorter waiting time to higher-priority tasks as they have shorter periods (and implicit deadlines) under RM. On the other hand, if the difference in task periods is relatively small, fair waiting time for tasks is likely to improve overall schedulability. There has been a long debate on priority-based and FIFO ordering in the literature. Further discussion on this issue is beyond the scope of the paper, and we leave the extension of the GPU server with FIFO ordering as part of future work.

In summary, the server-based approach outperforms the synchronization-based approach using MPCP and FMLP+ in most cases where realistic parameters are used. Specifically, the benefit of the server-based approach is significant when the percentage of GPU-using tasks is high, the number of GPU segments per tasks is large. However, we do find that the server-based approach does not dominate the synchronization-based approach. The synchronization-based approach may result in better schedulability than the server-based approach when the GPU server overhead or the ratio of miscellaneous operations in GPU segments is beyond the range of practical values (under MPCP and FMLP+) and the range of task period is relatively small (under FMLP+).

\section{Conclusions}
\label{conclusions}

In this paper, we have presented a server-based approach for predictable CPU access control in real-time embedded and cyber-physical systems. It is motivated by the limitations of the synchronization-based approach, namely busy-waiting and long priority inversion. By introducing a dedicated server task for GPU request handling, the server-based approach addresses those limitations, while ensuring the predictability and analyzability of tasks. The implementation and case study results on an NXP i.MX6 embedded platform indicate that the server-based approach can be implemented with acceptable overhead and performs as expected with a combination of CPU-only and GPU-using tasks. Experimental results show that the server-based approach yields significant improvements in task schedulability over the synchronization-based approach in most cases. Other types of accelerators such as DSPs and FPGAs can also benefit from our approach.

Our proposed server-based approach offers several interesting directions for future work. First, while we focus on a single GPU in this work, the server-based approach can be extended to multi-GPU and multi-accelerator systems. One possible way is to create a GPU server for each of the GPUs and accelerators, and allocating a subset of GPU-using tasks to each server. The optimization and dynamic adjustment of servers considering task allocation, overhead, and access duration is worthy of investigation. Secondly, the server-based approach can facilitate an efficient co-scheduling of GPU kernels. For instance, the latest NVIDIA GPU architectures can schedule multiple GPU kernels concurrently only if they belong to the same address space~\cite{NVIDIA_Pascal}, and the use of the GPU server satisfies this requirement, just as MPS~\cite{NVIDIA_MPS} does. Thirdly, as the GPU server has a central knowledge of all GPU requests, other features like GPU fault tolerance and power management can be developed. Lastly, the server-based approach can be extended to a virtualization environment and compared against virtualization-aware synchronization protocols~\cite{Kim_RTSS14}. We plan to explore these topics in the future.

%\scriptsize
\footnotesize
%\small
\bibliographystyle{abbrv}
\bibliography{paper}  % sigproc.bib is the name of the Bibliography in this case

\end{document}